\documentclass[11pt]{article}
\usepackage[margin=1in]{geometry}

\usepackage{comment}
\usepackage{graphicx, wrapfig}
\usepackage{subcaption}
\usepackage{enumitem}

\usepackage[utf8]{inputenc} %
\usepackage[T1]{fontenc}    %

\usepackage{amssymb,amsmath,amsthm,mathrsfs,bbm,mathtools}
\usepackage{amsfonts}       %
\usepackage{dsfont}
\usepackage{amsbsy}
\usepackage{bm}

\usepackage{nicefrac}       %
\usepackage{microtype}      %
\usepackage{etoolbox}
\usepackage{graphicx}
\usepackage{psfrag}
\usepackage{booktabs}       %
\usepackage{wrapfig}

\usepackage{mdframed}
\usepackage{enumitem}
\usepackage{setspace}
\usepackage{multirow}
\usepackage{colonequals}

\usepackage[usenames,dvipsnames]{xcolor}
\usepackage{url}   
\usepackage{cite}          %
\usepackage{hyperref}
\hypersetup{
    colorlinks=true,
    linkcolor=blue,
    filecolor=magenta,
    urlcolor=cyan,
    citecolor=blue
}

\allowdisplaybreaks

\newtheorem{proposition}{Proposition}

\makeatletter
\newtheorem*{rep@theorem}{\rep@title}
\newcommand{\newreptheorem}[2]{%
\newenvironment{rep#1}[1]{%
 \def\rep@title{#2 \ref{##1}}%
 \begin{rep@theorem}}%
 {\end{rep@theorem}}}
\makeatother
\newreptheorem{conj}{Conjecture}

\theoremstyle{definition}
 
\newtheorem{definition}{Definition} 
 
\newtheorem{remark}{Remark} 

\newtheoremstyle{tasksty}%
  {3pt}%
  {3pt}%
  {\em}%
  {3pt}%
  {\color{MidnightBlue}\sffamily\mdseries\upshape}%
  {:}%
  { }%
  {}%

\theoremstyle{tasksty}

\newtheoremstyle{sGoal}%
  {3pt}%
  {3pt}%
  {\em}%
  {3pt}%
  {\color{BrickRed}\sffamily\mdseries\upshape}%
  {:}%
  { }%
  {}%

\theoremstyle{sGoal}

\theoremstyle{plain}
\newtheorem{theorem}{Theorem}

\newcommand{\ee}{\end{equation}}

\newcommand{\Reals}{\mathbb{R}}

\newcommand{\calD}{\mathcal{D}}

\newcommand{\calF}{\mathcal{F}}

\newcommand{\calH}{\mathcal{H}}

\newcommand{\calY}{\mathcal{Y}}

\newcommand{\bx}{\boldsymbol{x}}

\newcommand{\by}{\boldsymbol{y}}

\DeclareMathOperator*{\argmin}{\arg\!\min}
\DeclareMathOperator*{\argmax}{\arg\!\max}

\newcommand{\normEuc}[1]{\| #1 \|_2}

\newcommand{\ones}{\boldsymbol{1}}

\newcommand{\ExpVal}[2]{\mathbb{E}\left[ #2 \right]}

\newcommand{\EE}[1]{\ExpVal{}{#1}}

\definecolor{light-gray}{gray}{.90}
\definecolor{aliceblue}{rgb}{0.94, 0.97, 1.0}
\definecolor{airforceblue}{rgb}{0.36, 0.54, 0.66}
\definecolor{bleudefrance}{rgb}{0.19, 0.55, 0.91}
\definecolor{cerulean}{rgb}{0.0, 0.48, 0.65}

\newmdenv[%
  linewidth =0pt,%
  fontcolor=bleudefrance,
  innertopmargin=2pt,
  innertopmargin=2pt,
  leftmargin = 0pt,
  rightmargin = 0pt,
  innerleftmargin = 2pt,
  innerrightmargin = 2pt,
  skipabove = 6pt,%
  skipbelow = 6pt
]{RQ}

\newmdenv[%
  backgroundcolor=aliceblue, %
  linewidth = 2pt,%
  skipabove = 10pt,%
  skipbelow = 10pt,
  pstrickssetting={linestyle=dashed,},
  linecolor=airforceblue,
  middlelinewidth=2pt
]{TODO}

\usepackage[mathcal]{euscript}

\usepackage[utf8]{inputenc} %
\usepackage[T1]{fontenc}    %
\usepackage{hyperref}       %
\usepackage{url}            %
\usepackage{booktabs}       %
\usepackage{amsfonts}       %
\usepackage{nicefrac}       %
\usepackage{microtype}      %
\usepackage{xcolor}         %
\usepackage{algorithm}
\usepackage{algorithmic}
\usepackage[numbers,sort&compress]{natbib}
\newcommand{\YelpModel}{\texttt{YelpLLM }}

\newcommand{\hq}{\textsf{HQ}}
\newcommand{\MC}{\textsf{MC}}
\newcommand{\amortized}{\textsf{Amor}}
\newcommand{\selective}{\textsf{SE}}
\newcommand{\ig}{\textsf{IG}}
\newcommand{\LambdaName}{combination function }

\newcommand{\uncert}{s_{h}}

\begin{document}

\title{\textbf{Selective Explanations}}

\author{Lucas Monteiro Paes${}^{1}$~\footnote{Correspondence to Lucas Monteiro Paes (lucaspaes@g.harvard.edu).}~, Dennis Wei${}^{2}$, Flavio P. Calmon${}^{1}$}

\date{${}^{1}$Harvard University \\ ${}^{2}$IBM Research}

\maketitle
\begin{abstract}
Feature attribution methods explain black-box machine learning (ML) models by assigning importance scores to input features. These methods can be computationally expensive for large ML models. To address this challenge, there has been increasing efforts to develop \emph{amortized explainers}, where a machine learning model is trained to predict feature attribution scores with only one inference. Despite their efficiency,  amortized explainers can produce inaccurate predictions and misleading explanations. In this paper, we propose \emph{selective explanations}, a novel feature attribution method that (i) detects when amortized explainers generate low-quality explanations and (ii) improves these explanations using a technique called \emph{explanations with initial guess}.  Our selective explanation method allows practitioners to specify the fraction of samples that receive explanations with initial guess, offering a principled way to bridge the gap between amortized explainers  and their high-quality counterparts.
\end{abstract}

\section{Introduction}
Large black-box models are increasingly used to support decisions in applications ranging from online content moderation \citep{GPT4Content},  hiring  \citep{ghosh2023jobrecogpt}, and medical diagnostics \citep{wang2023chatcad}. In such high-stakes settings, the need for explaining ``why'' a model produces a given output has led to a growing number of perturbation-based \emph{feature attribution} methods \citep{Lundberg_shap, tulio_lime, paes2024multilevel, miglani2023captum, Chen2019LSTreeMI, Zeiler_CNNs}. Broadly speaking, these methods use input perturbations to assign numerical values to each input feature a model uses, indicating their influence on predictions. %
They are widely adopted in part because they work in the black-box setting with access only to model outputs (i.e., no gradients). However, existing feature attribution methods can be prohibitively expensive for the large models used in the current machine learning landscape (e.g., language models with billions of parameters) since they require a significant number of inferences for each individual explanation.

Recent literature has introduced two main strategies for speeding up feature attribution for large models: (i) employing Monte Carlo methods to approximate explanations with fewer computations \citep{Lundberg_shap, tulio_lime, covert_kernel, Mitchell_svs}, and (ii) adopting an \emph{amortized} approach, training a separate model to ``mimic'' the outputs of a reference explanation method  \citep{jethani2022fastshap, covert2024stochastic, Yang2023EfficientSV, cxplain, Chuang2023CoRTXCF, schwarzenberg-etal-2021-efficient}.
Monte Carlo approximations can yield high-quality explanations but may converge slowly, limiting their practicality for large datasets.
Amortized explainers, in turn, require only one inference per explanation, making them efficient for large black-box models and datasets. However, as demonstrated in Figure \ref{fig:comparing_explanations}, amortized explainers can occasionally produce  diverging explanations from the reference explainer used to train them.%

We propose \emph{selective explanation}, a method that bridges Monte Carlo and amortized explanations. By training a model that ``learns to select'' which method should be applied to each input, our selective explanation method can produce higher-quality explanations than amortized explainers at a significantly lower average computational cost than Monte Carlo-based approaches. The key idea behind the selective explanation method is to apply Monte Carlo explanations only to points that would receive low-quality explanations from the amortized explainer; see Figure \ref{fig:diagram} for the workflow of selective explanations. 

The ideas of predicting selectively and providing recourse with a more accurate but expensive method have been explored in classification and regression \citep{Rabanser2022SelectiveCV, selective_classification, el2010foundations, gangrade2021selective, geifman2019selectivenet}. To our knowledge, however, these ideas have not been applied to explanations. We make \textbf{two contributions} in this regard that are relevant for selective prediction more generally. (1) Selective prediction uses an \emph{uncertainty metric} to identify input points for which the predictor (the amortized explainer in our case) would produce low-quality outputs and recourse is needed. The high-dimensional nature of explanations requires us to develop new uncertainty metrics (Section~\ref{sec:selecting}) suitable for this setting. (2) Instead of providing recourse with a Monte Carlo explanation alone, as would be standard, we use an optimized method called \emph{explanations with initial guess} (Section~\ref{sec:recourse}) that combines amortized and Monte Carlo explanations, improving explanation quality beyond that of either explanation alone.

Our \textbf{overall contribution} (3) is to combine (1) and (2) in the form of \emph{selective explanations}, providing explanations with initial guess to improve low-quality amortized explanations.
We validate our selective explanations approach on two language models as well as tabular datasets, demonstrating its ability to accurately detect low-quality explanations, enhance amortized explanations with even low-quality Monte Carlo explanations, and improve the worst explanations from the amortized model.

\begin{figure}[t]
  \centering
    \begin{subfigure}[b]{0.3\linewidth}
    \includegraphics[width=\linewidth]{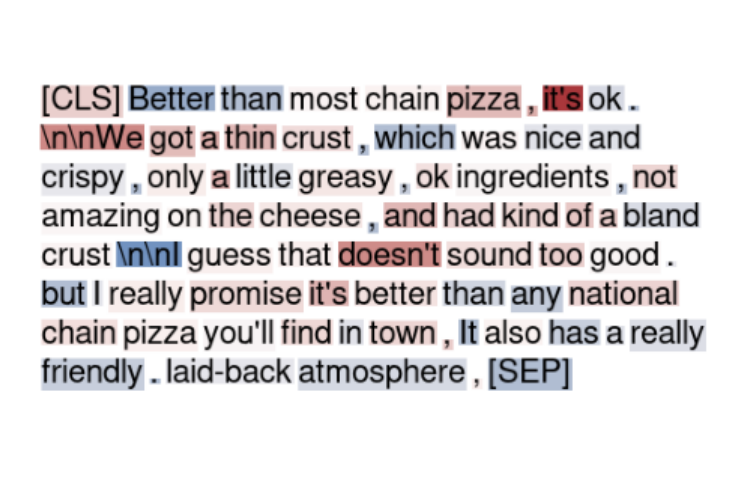}
    \caption{Amortized (MSE = 0.31)}
  \end{subfigure}
\begin{subfigure}[b]{0.3\linewidth}
    \includegraphics[width=\linewidth]{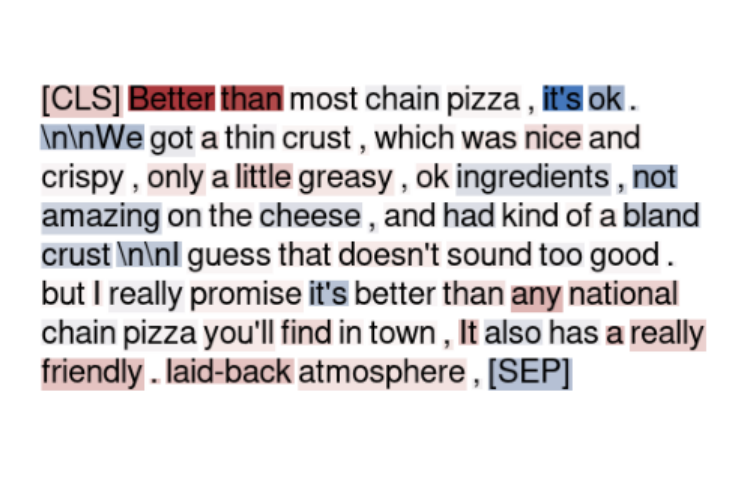}
    \caption{High-Quality (target)}
  \end{subfigure}
  \begin{subfigure}[b]{0.3\linewidth}
    \includegraphics[width=\linewidth]{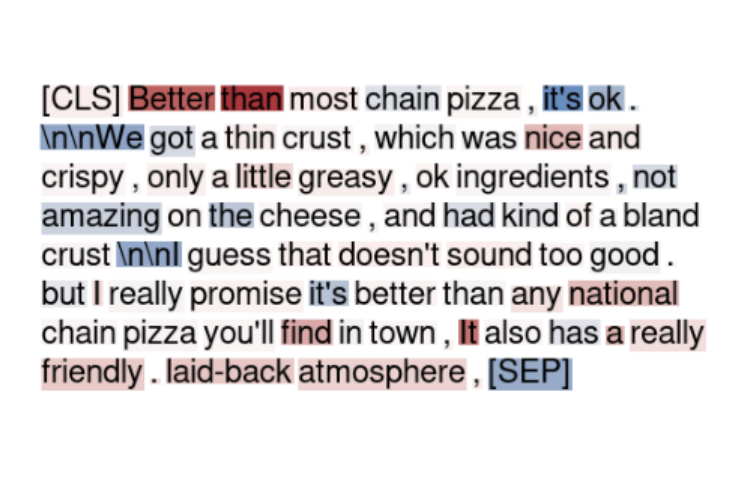}
    \caption{Selective (MSE = 0.07)}
  \end{subfigure}
  \caption{Amortized explainer (a) compared with a high-quality explainer (b) and our selective explanation method (c). All methods flag inputs that attribute why \YelpModel predicted \texttt{Negative Review} in the example. We observe that both high-quality and selective explanations attribute "not amazing" for the negative review (blue), while the amortized explainer misses this term. Similarly, the amortized explainer incorrectly the expression ``Better than."}
  \label{fig:comparing_explanations}
\end{figure}

\begin{figure}[t]
    \centering
    \includegraphics[width=0.9\linewidth]{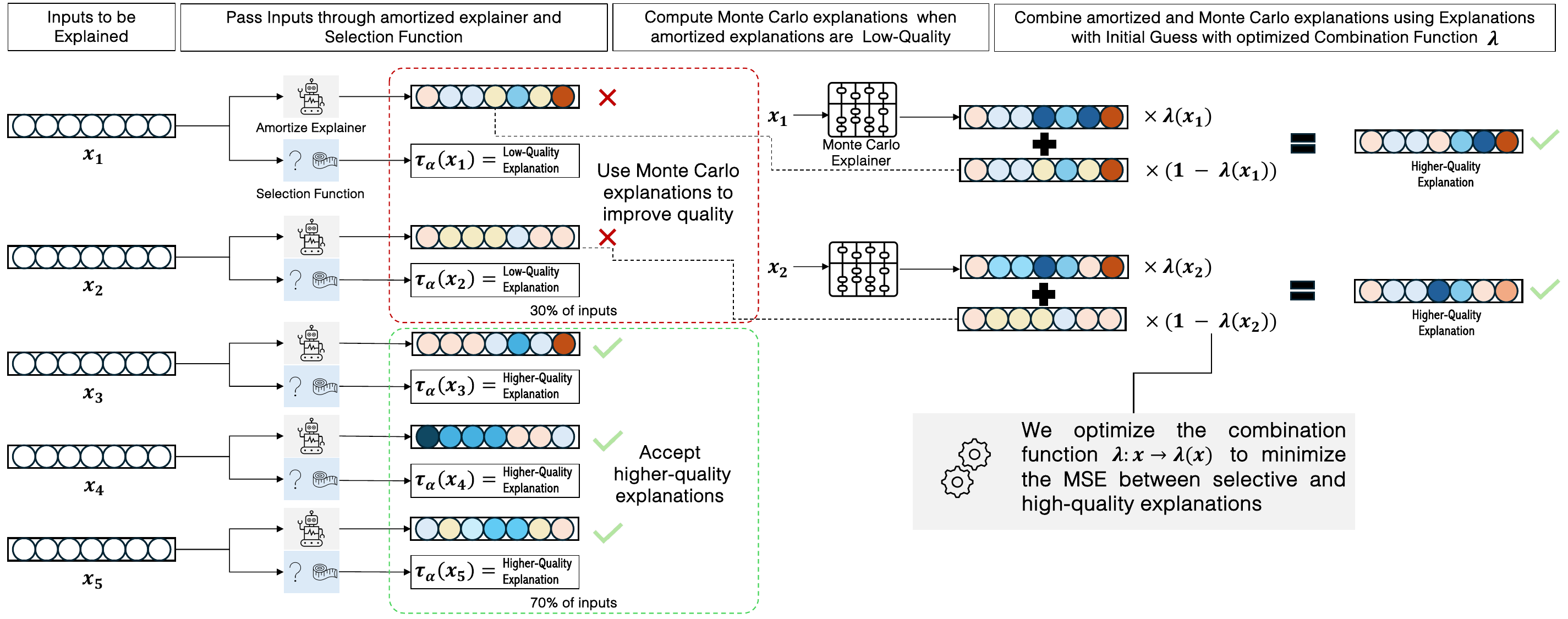}
    \caption{Workflow of selective explanations.}
    \label{fig:diagram}
\end{figure}

\section{Problem Setup \& Background}
\label{sec:background}

We aim to explain the predictions of a fixed probabilistic black-box model $h$ that predicts $h(\bx) = (h_1(\bx), ..., h_{|\calY|}(\bx))$ and outputs $\argmax_{j \in \calY} h_j(\bx) \in \calY$ using a vector of features $\bx = (x_1, ..., x_d) \in \mathbb{R}^{d}$.
The user specifies an output of interest $\by \in \calY$ (usually $\by = \argmax_{j \in \calY} h_j(\bx)$) and our goal is to explain \emph{Why would $h$ output $\by$ for a given $\bx$?} 
We consider a dataset $\calD = \{(\bx_i, \by_i)\}_{i=1}^{N}$ comprised of $N > 0$ samples divided into three parts: $\calD_{\texttt{train}}$ for training $h$ and the explainers, $\calD_{\text{cal}}$ for calibration and validation, and $\calD_{\texttt{test}}$ for testing. Thus, $\calD = \calD_{\texttt{train}} \cup \calD_{\text{cal}} \cup \calD_{\texttt{test}}$.
Moreover, for a subset $S = \{i_1, ..., i_{|S|}\} \subset [d]$ we write $\bx_{S} \triangleq (x_{i_1}, ..., x_{i_{|S|}})$.

\textbf{Feature Attribution Methods,} also called \emph{explainers}, are functions $\Reals^d \times \mathcal{Y}\to \Reals^d$ that assess the importance of each feature for the model's ($h$) prediction to be $\by$ for a given input vector $\bx$. 
We consider three types of explainers: 
\begin{enumerate}[label=(\roman*)]
    \item \textbf{High-quality explainers} that use a large number of computations to provide explanations (e.g., SHAP with $2^d$ inferences from model $h$) \citep{Lundberg_shap, tulio_lime}, denoted by $\hq(\bx, \by)$;
    \item  \textbf{Monte Carlo explainers} that approximate  high-quality explainers  using $n$ inferences from model $h$ per explanation \citep{Lundberg_shap, Mitchell_svs}, denoted by $\MC^{n}(\bx, \by)$; 
    \item \textbf{Amortized explainer} trained to approximate the high-quality explanations using only one inference \citep{covert2024stochastic, Yang2023EfficientSV}, denoted by $\amortized(\bx, \by)$.
\end{enumerate}

 We measure the difference between two competing explanations using a loss (or distortion) function $\ell: \mathbb{R}^d \times \mathbb{R}^d \rightarrow \mathbb{R}$, e.g., mean square error (MSE). The goal of selective explanations ($\selective$) is to approximate high-quality explanations while minimizing the number of computations, i.e., to minimize $\left|\left| \selective(\bx, \by) - \hq(\bx, \by) \right|\right|_2^2$. However, computing high-quality explanations for large models $h$ can be prohibitively expensive. To address this issue, we define \emph{selective explainers} below.

\begin{definition}[\textbf{Selective Explainer}] %
\label{def:selective_explanations}
For a given model $h$, an amortized explainer $\amortized$, a Monte Carlo explainer $\MC^n$, a \emph{\LambdaName}  $\lambda_h: \Reals^d \to \Reals$, and a \emph{selection function} $\tau_{\alpha}:\Reals^d\to \{0,1\}$ (parametrized by $\alpha$), we define the \emph{selective explainer} $\selective(\bx, \by)$ as
\begin{align}
    \selective(\bx, \by) \triangleq
    \begin{cases}
        \amortized(\bx, \by)&, \text{ if } \tau_{\alpha}(\bx) = 1,\\
        \lambda_h(\bx) \amortized(\bx, \by) + (1 - \lambda_h(\bx))\MC^n(\bx, \by)&, \text{ if } \tau_{\alpha}(\bx) = 0.\\
    \end{cases}
    \label{eq:selective_explanation}
\end{align}
\end{definition}

When $\tau_{\alpha} = 0$, selective explanations output \emph{explanations with initial guess} (Definition \ref{def:explanation_initial_guess}). Explanations with initial guess optimally linearly combine amortized and Monte Carlo explanations to leverage information from both and provide higher-quality explanations than either explainer alone. Selective explanations heavily depend on three objects that we define in this work: (i) an uncertainty metric (Section \ref{sec:selecting}), (ii) a selection function (Section \ref{sec:selecting}), and (iii) a combination function (Section \ref{sec:recourse}).

\begin{itemize}
    \item \textbf{Uncertainty metrics} ($\uncert$) output the likelihood of the amortized explainer producing a low-quality explanation for an input. Lower $\uncert(\bx)$ indicates a higher-quality explanation for $\bx$. We propose two uncertainty metrics: Deep and Learned Uncertainty (Section \ref{sec:selecting}).
    
    \item \textbf{Selection function} ($\tau_{\alpha}$) is a binary rule that outputs 1 for high-quality amortized explanations and 0 for low-quality ones based on the uncertainty metric. We define $\tau_{\alpha}$ to ensure a fraction $\alpha$ of inputs receive amortized explanations. Smaller $\alpha$ implies higher-quality selective explanations but also more computations (Section \ref{sec:selecting}).
    
    \item \textbf{Combination function} ($\lambda_h$) optimally linearly combines amortized and Monte Carlo explanations to minimize MSE from high-quality explanations (Theorem \ref{thm:optimal_lambda}). We propose explanations with initial guess and fit $\lambda_h$ to optimize their quality (Section \ref{sec:recourse}).
\end{itemize}

\begin{algorithm}[b]
\caption{Building a Selective Explainer}
\label{alg:selective_explainer}
\begin{algorithmic}[1]
\REQUIRE{
Datasets: $\mathcal{D}_{\texttt{train}}$, $\mathcal{D}_{\texttt{cal}}$.
Explainers: $\amortized$,
$\MC^n$, $\MC^{n'}$.
Coverage: $\alpha$.
}

\ENSURE
Selection function: $\tau_{\alpha}$.
\LambdaName: $\lambda_h$.

\STATE Fit the uncertainty metric $\uncert$ using $\calD_{\texttt{train}}$, $\amortized$, and $\MC^n$ (using \eqref{eq:DeepUncertainty} or \eqref{eq:UncLearn})

\STATE Compute $t_{\alpha}$ using  $\calD_{\texttt{cal}}$ \eqref{eq:selecting_threshold}

\STATE Define the selection function $\tau_{\alpha}$ using $\uncert$ and $t_{\alpha}$ \eqref{eq:decision_rule}

\STATE Define bins $Q_i = [t_{\alpha_i}, t_{\alpha_{i+1}})$ for partition $\alpha_i = \frac{i-1}{k}$ for $i \in [k+1]$ \eqref{eq:quantile_definition}

\STATE For $i \in [k+1]$ Compute $\lambda_i$ as in \eqref{eq:lambda_aprox_definition} using $\mathcal{D}_{\texttt{cal}}$, $\amortized$,
$\MC^n$, and $\MC^{n'}$.

\STATE Define $\lambda_h(\bx) = \sum_{i = 1}^{k+1} \lambda_i\ones[\uncert(\bx) \in Q_i] $ as in \eqref{eq:quantile_definition}

\RETURN  $\tau_{\alpha}$, $\lambda_h(\bx)$
\end{algorithmic}
\end{algorithm}

%
%
%
%

%
%

%

%
Algorithm \ref{alg:selective_explainer} describes the procedure to compute the uncertainty metric, selection function, and \LambdaName using the results we describe in Section \ref{sec:selecting} and \ref{sec:recourse}.
Although selective explanations can be applied to any feature attribution method, we focus on Shapley values since they are widely used and most amortized explainers are tailored for them \citep{jethani2022fastshap, Yang2023EfficientSV, covert2024stochastic}.
We discuss how selective explanations can be applied to LIME and provide more details on feature attribution methods in Appendix~\ref{apx:additional)explanations}.
Next, we describe specific feature attribution methods that we use as building blocks for selective explainers of the form \eqref{eq:selective_explanation}.

\textbf{Shapley Values (SHAP)}  \citep{Lundberg_shap}
is a \textbf{high-quality} explainer
that attributes a value $\phi_i$ for each feature $x_i$ in $\bx = (x_1, ..., x_d)$ which is the marginal contribution of feature $x_i$ if the model was to predict $\by$
\begin{equation}
    \phi_i(\bx, \by) = \frac{1}{d} \sum_{S \subset [d]/\{i\}}  {d - 1 \choose |S|}^{-1} \left( h_{\by}(\bx_{S \cup \{i\}}) - h_{\by}(\bx_{S}) \right).
    \label{eq:shapley_values}
\end{equation}
SHAP has several desirable properties and is widely used. However, as \eqref{eq:shapley_values} indicates, computing Shapley values and the attribution vector $\hq(\bx, \by) = (\phi_1(\bx, \by), ..., \phi_d(\bx, \by))$ requires $2^d$ inferences from $h$, making SHAP impractical for large models where inference is costly.
This has motivated several approximation methods for SHAP, discussed next.

\textbf{Shapley Value Sampling (SVS)} \citep{Mitchell_svs} is a \textbf{Monte Carlo} explainer that approximates SHAP by restricting the sum in \eqref{eq:shapley_values} to $m$ uniformly sampled permutations of features performing $n = md +1$ inferences.
We denote SVS that samples $m$ feature permutations by SVS-$m$.

\textbf{Kernel Shap (KS)} \citep{Lundberg_shap} is a \textbf{Monte Carlo} explainer that approximate Shapley values using the fact that SHAP can be computed by solving a weighted linear regression problem using $n$ input perturbations resulting in $n$ inferences.
We refer to Kernel Shap using $n$ inferences as KS-$n$. 

\textbf{Stochastic Amortization} \citep{covert2024stochastic} is a \textbf{Amortized} explainer that uses noisy Monte Carlo explanations to learn high-quality explanations. \citet{covert2024stochastic} trained an amortized explainer in a model class $\calF$ (multilayer perceptrons) $\amortized \in \calF$  to take $(\bx, \by)$ and predicts an explanation $\amortized(\bx, \by) \approx \hq(\bx, \by)$ by minimize the $L_2$ norm from Monte Carlo explanations $\MC^{n}(\bx, \by)$.
Specifically, the amortized explainer is given by
\begin{equation}
    \amortized \in \argmin_{f \in \calF} \sum_{(\bx, \by) \in \calD_{\text{train}}} \normEuc{ f(\bx, \by) - \MC^{n}(\bx, \by)}^{2}.
    \label{eq:training_amortized}
\end{equation}

\textbf{Amortized Shap for LLMs} \citep{Yang2023EfficientSV}
is a \textbf{Amortized} explainer similar to stochastic amortization but tailored for LLMs. 
\citet{Yang2023EfficientSV} train a linear regression on the LLM embeddings $[e_1(\bx), ..., e_{|\bx|}(\bx)]$ to minimize the $L_2$ norm from Monte Carlo explanations $\MC^{n}(\bx, \by)$ and define the amortized explainer as
$
    \amortized(\bx, \by) = (W_{\by}e_1(\bx) + b_{\by}, ..., W_{\by}e_{|\bx|}(\bx) + b_{\by}),
$
$W_{\by}$ is a matrix and $b_{\by} \in \Reals$.

We use stochastic amortization to produce amortized explainers for tabular datasets and Amortized Shap for LLMs to produce explainers for LLM predictions. Both explainers are trained using SVS-12 as $\MC^{n}$. High-quality and Monte Carlo explanations are computed using the Captum library \citep{captum_lib}.
\section{Selecting Explanations}
\label{sec:selecting}
In this section, we define key concepts for selective explainers: (i) uncertainty metrics $\uncert$ to quantify the likelihood of an explanation being low-quality and (ii) selection functions ($\tau_{\alpha}$) to predict when amortized explanations are high-quality based on the value of an uncertainty metric.

\paragraph{Uncertainty Metrics for High-Dimensional Regression:} An uncertainty metric is a function tailored for the model $h$ that takes $\bx$ and outputs a real number $\uncert(\bx)$ that encodes information about the uncertainty of the model $h$ in the prediction for $\bx$. Generally, if $\uncert(\bx) < \uncert(\bx')$ then the model is more confident about the prediction $h(\bx)$ than $h(\bx')$ \citep{selective_classification, Rabanser2022SelectiveCV}.
Existing uncertainty metrics cater to (i) classification \citep{Rabanser2022SelectiveCV, selective_classification, el2010foundations, gangrade2021selective, geifman2019selectivenet} and (ii) one-dimensional regression \citep{zaoui2020regression, shah2022selective, geifman2019selectivenet, jiang2020risk}, but none specifically address high-dimensional regression -- which is our case of interest ($d$-dimensional explanations).
Next, we propose two uncertainty metrics tailored to high-dimensional outputs: (i) Deep uncertainty and (ii) Learned uncertainty.

\textbf{Deep Uncertainty} is inspired by deep ensembles \citep{lakshminarayanan2017simple}, a method that uses an ensemble of models to provide confidence intervals for the predictions of one model.
We run the training pipeline for the amortized explainer described in \eqref{eq:training_amortized} $k$ times, each with a different random seed, resulting in $k$ different amortized explainers $\amortized^1, ..., \amortized^k$.
We define the deep uncertainty as
\begin{equation}
    \uncert^{\texttt{Deep}}(\bx) \triangleq \frac{1}{dk} \sum_{i = 1}^{d} \text{Var}\left(\amortized^1(\bx)_{i}, ..., \amortized^k(\bx)_{i}\right).
    \label{eq:DeepUncertainty}
\end{equation}
Here, $\text{Var}\left(a_1, ..., a_k \right)$ is the variance of the sample $\{a_1, ..., a_k\}$ and $\amortized^j(\bx)_{i}$ indicates the $i$-th entry of the feature attribution vector $\amortized^j(\bx)$.
Hence, deep uncertainty is the average (across entries) of the variance (across all trained amortized explainers) for the predicted attributions.

If the deep uncertainty for a point $\bx$ is zero, then the amortized explainers produce the same feature attribution.
On the other hand, if the deep uncertainty is high, then the feature attributions vary widely across the amortized explainers.
Intuitively, the points with a higher deep uncertainty are more affected by a random seed change, implying more uncertainty in the explanation.%

\textbf{Learned Uncertainty} uses data to predict the amortized explainer uncertainty at an input point $\bx$.
We choose $\ell$ (the loss function) between two explanations to be MSE. The learned uncertainty metric is a function in the class $\calF$ (multilayer perceptron in our experiments) such that
\begin{equation}
    \uncert^{\texttt{Learn}} \in \argmin_{s \in \calF} \sum_{(\bx, \by) \in \calD_{\text{train}}} \left| s(\bx) - \ell\left(\amortized(\bx; \by), \MC^n(\bx; \by)\right)\right|^2.
    \label{eq:UncLearn}
\end{equation}
Ideally, instead of using the Monte Carlo explanation $\MC^n$ as the reference in \eqref{eq:UncLearn}, we would like to use high-quality explanations, i.e., $\ell\left(\amortized(\bx; \by), \hq(\bx; \by)\right)$. However, these computationally expensive explanations are usually not available.
Thus, we resort to using Monte Carlo explanations.

For large language models, the textual input $\bx$ is encoded in a sequence of token embedding $[e_1(\bx), ..., e_{|\bx|}(\bx)]$ such that $e_{i}(\bx) \in \mathbb{R}^{d}$ for $i \in [|\bx|]$.
In this case, we use the mean (i.e., ``mean-pooling'') of the token embeddings to train the learned uncertainty metric instead of $\bx$.

We analyze the performance of the proposed uncertainty metrics in Section \ref{ssec:coverage_vs_performance}, showing that it can be used to detect low-quality explanations from the amortized explainer. Our results indicate that these functions closely approximate the best possible uncertainty measure -- the Oracle with knowledge of high-quality explanations (Figure \ref{fig:coverage_MSE}).
Next, we define the selection function that allows practitioners to set a coverage (percentage of points) $\alpha$ that will receive amortized explanations.

\paragraph{Selection functions:} a selection function is the binary qualifier ($\tau_{\alpha}$) that thresholds the uncertainty metric by $t_{\alpha} \in \mathbb{R}$ given by
\begin{align}
    \tau_{\alpha}(\bx) \triangleq \begin{cases}
        1 & \text{ if } \uncert(\bx) \leq t_{\alpha} \texttt{ (high-quality explanations)}\\
        0 & \text{ if } \uncert(\bx) > t_{\alpha} \texttt{ (low-quality explanations)}
    \end{cases}.
    \label{eq:decision_rule}
\end{align}
Intuitively, $t_{\alpha}$ is the maximum uncertainty level tolerated by the user.
In practice, if the output of the selection function is $1$ (high-quality explanation), we use the explanations from the amortized model; if the output of the selection function is $0$ (low-quality explanation), we use explanations with initial guess (see Definition \ref{def:explanation_initial_guess} bellow) to improve the explanation provided to the user.
The threshold $t_{\alpha}$ is chosen to be the $\alpha$-quantile of the uncertainty metric to ensure that at least a fraction $\alpha$ of points receive a computationally cheap explanation -- we call $\alpha$ the \emph{coverage}.
Specifically, given $\alpha$, we calibrate $t_{\alpha}$ in the calibration dataset $\calD_{\texttt{cal}}$ and compute it as
\begin{equation}
    t_{\alpha} \triangleq \min_{t \in \Reals} t, \text{  such that } \Pr_{\texttt{cal}}[ \uncert(\bx) \leq t] \geq \alpha,
    \label{eq:selecting_threshold}
\end{equation}
where $\Pr_{\texttt{cal}}$ is the empirical distribution of the calibration dataset. For discussions on selecting coverage with guarantees on the number of inferences for selective explanations, see Appendix \ref{sec:coverage_for_budget}.

\begin{remark}
    \label{rem:theory_guarantee}
    A property of selective predictions \citep{selective_classification}, that is transferred to selective explanations is that it is possible to control the explainer's performance via the threshold $t_{\alpha}$ with guaranteed performance but without providing predictions for all points. This result is displayed in Figure \ref{fig:coverage_MSE}. 
\end{remark}

\section{Explanations with Initial Guess}
\label{sec:recourse}
In the previous section, we introduced methods to detect points likely to receive low-quality explanations from amortized explainers. This raises the question: \emph{How can we improve the explanations for these points?} One approach is to simply use Monte Carlo ($\MC$) explanations instead of amortized explanations. However, this ignores potentially valuable information already computed by the amortized explainer. In this section, we propose a more effective solution called \emph{explanations with initial guess}, which combines amortized and Monte Carlo explanations to improve explanation quality.

\textbf{Explanation with Initial Guess}
uses an optimized linear combination of the amortized explanation with a more computationally expensive method -- the Monte Carlo explainer -- to improve the quality of the explanation.
We formally define \emph{explanations with initial guess} next.

\begin{definition}[Explanation with Initial Guess]
\label{def:explanation_initial_guess}
Given a Monte Carlo explainer $\MC^n(\bx, \by)$, and a \LambdaName  $\lambda_h : \mathbb{R}^d \rightarrow \Reals$ that reflects the quality of the amortized explanation $\amortized$, we define the explanation with initial guess as
    \begin{equation}
         \ig(\bx, \by) \triangleq \lambda_h(\bx) \amortized(\bx, \by) + (1 - \lambda_h(\bx))\MC^n(\bx, \by).
         \label{eq:initial_guess}
\end{equation}
\end{definition}
Recall that when $\tau_{\alpha}(\bx) = 0$, selective explanations use the explanation with initial guess \eqref{eq:selective_explanation} to improve low-quality amortized explanations, i.e.,
$
    \selective(\bx, \by) = \ig(\bx, \by).
$

Defining explanations with initial guess as the linear combination between the amortized and the Monte Carlo explanations is inspired by the literature on shrinkage estimators \citep{lemmer1981ordinary, lemmer1981note} that use an initial guess ($\amortized(\bx, \by)$ in our case) to improve the estimation MSE in comparison with only using the empirical average (a role played by $\MC^n(\bx, \by)$ in our case).
Next, we tune $\lambda_h$ to minimize the MSE from high-quality explanations.

\paragraph{Optimizing the Explanation Quality:} 
Our goal is for explanations with initial guess to approximate the high-quality explanations from $\hq$, i.e., $||  \ig(\bx, \by) -  \hq(\bx, \by)||$ to be minimized.
To achieve this, we optimize the function $\lambda_h$ as follows.

First, since high-quality explanations $\hq$ are unavailable, we use another Monte Carlo explanation $\MC^{n'}$ that closely approximates $\hq$. $\MC^{n'}$ is different from $\MC^{n}$ and potentially more computationally expensive. Importantly, $\MC^{n'}$ is only needed beforehand when computing $\lambda_h$, not at prediction time.
In our experiments, we use SVS-12 for $\MC^{n'}$.

Second, we quantize the range of the uncertainty metric $\uncert$ into bins to aggregate points with similar uncertainty and define the bins $Q_i$ by a partition $0=\alpha_1 < \alpha_2 < ... < \alpha_m = 1$ of $[0, 1]$: 
\begin{equation} 
Q_{i} \triangleq [t_{\alpha_i}, t_{\alpha_{i+1}}), \ \ \forall i \in [m-1] \label{eq:quantile_definition} 
\end{equation} 
where $t_{\alpha_i}$ is defined as in \eqref{eq:selecting_threshold}.
We then define the \LambdaName to be
\begin{equation}
    \lambda_h(\bx) = \lambda_i \text{ if } \uncert(\bx) \in Q_i,
\end{equation}
$\lambda_h$ is chosen to optimize the explanation-quality for points with similar uncertainty, $\lambda_i$ is given by: 
\begin{equation} 
\lambda_i \triangleq \argmin_{\lambda \in \mathbb{R}} \sum_{\substack{ (\bx, \by) \in \calD_{\texttt{cal}} \\ \uncert(\bx) \in Q_i}} \left|\left| \ig(\bx, \by) - \MC^{n'}(\bx, \by) \right|\right|_2^2. 
\label{eq:lambda_problem} 
\end{equation}

We only compute $\lambda_i$ once per bin we provide explanations with initial guess \eqref{eq:initial_guess}, i.e., when $\tau_{\alpha}(\bx) = 0$.

Theorem \ref{thm:optimal_lambda} provides a closed-form solution for $\lambda_i$.
\begin{theorem}[Optimal $\lambda_h$]
\label{thm:optimal_lambda}
Let $0=\alpha_1 < \alpha_2 < ... < \alpha_m = 1$ and define $Q_i$ as in \eqref{eq:quantile_definition}.
Then the solution to the optimization problem in \eqref{eq:lambda_problem} is given by
\begin{equation}
    \lambda_i = \frac{\sum_{\substack{ (\bx, \by) \in \calD_{\texttt{cal}} \\ \uncert(\bx) \in Q_i}} \langle \MC^{n}(\bx, \by) - \MC^{n'}(\bx, \by), \MC^{n}(\bx, \by) - \amortized(\bx, \by) \rangle}{\sum_{\substack{ (\bx, \by) \in \calD_{\texttt{cal}} \\ \uncert(\bx) \in Q_i}} \left|\left| \amortized(\bx, \by) - \MC^n(\bx, \by)\right|\right|_2^2}.
    \label{eq:lambda_aprox_definition}
\end{equation}
\end{theorem}

The range of uncertainty functions is \textbf{quantized} for two main reasons.
First, the uncertainty metric $\uncert$ encodes the amortized explainer's uncertainty for each point $\bx$. This uncertainty quantification should be reflected in the choice of $\lambda_h$. Quantizing the range of $\uncert$ allows us to group points with similar uncertainty levels and optimize $\lambda_h$ for each group separately.
Second, quantizing the range of $\uncert$ enables us to have multiple point per bin $Q_i$ allowing us to compute $\lambda_i$ to minimize the MSE in each bin.

We use the \textbf{Monte Carlo} explainer $\MC^{n'}$ because: (i) as mentioned above, we assume we don't have access to high-quality explanations due to there computational cost and (ii) even when using this Monte Carlo explainer, we show that in all bins $\lambda_i$ approximates well the optimal \LambdaName computed assuming access to high-quality explanations from $\hq$ defined as

$$
    \lambda^{\texttt{opt}}_i = \argmin_{\lambda \in [0, 1]} \sum_{\substack{ (\bx, \by) \in \calD_{\texttt{cal}} \\ \uncert(\bx) \in Q_i}} \left|\left|  \ig(\bx, \by) - \hq(\bx, \by) \right|\right|_2^2.
$$
Specifically, Theorem \ref{thm:lambda_are_close} shows that $\lambda_i \approx \lambda^{\texttt{opt}}_i$, Appendix \ref{apx:proofs} shows the formal version of the Theorem along with the proofs for all results in this section.

\begin{theorem}[\textbf{Informal} $\lambda_i \approx \lambda^{\texttt{opt}}_i$]
\label{thm:lambda_are_close}
If (i) $\MC^{n}$ is sufficiently different from the amortized explainer $\amortized$ and (ii) $\MC^{n'}$ approximates the high-quality explanations $\hq$ then $\lambda_i$ and $\lambda^{\texttt{opt}}_i$ are close with high-probability for all bins $Q_i$, i.e.,
\begin{equation*}
    |\lambda_i -\lambda^{\texttt{opt}}_i| \leq \epsilon \text{ with probability at least } 1 - e^{-C|Q_i|}. 
\end{equation*}
for a $C>0$ and $|Q_i|$ is the number of points in the validation dataset $\calD_{\texttt{cal}}$ that are in the bin $Q_i$.
\end{theorem}

\section{Experimental Results}
\label{sec:experiments}
This section analyzes the performance of selective explanations.
We show that (i) uncertainty metrics accurately identify low-quality explanations (Figure \ref{fig:coverage_MSE}), (ii) explanations with initial guess have a higher quality than amortized and Monte Carlo explanations (Figure \ref{fig:initial_guess_mse}), (iii) selective explanations improve the performance of the lowest-quality explanations (Figure \ref{fig:performance_in_bottom}), and (iv) selective explanations improve local fidelity (Figure \ref{fig:local_fidelity}).
We also (a) analyze how the quality of Monte Carlo explanations impact explanations with initial guess (Appendix \ref{apx:ablation_different_recourse}) and (b) show that selective explanations can be used to improve the inference vs. MSE trade-off of Monte Carlo explanations (Appendix \ref{apx:time_sharing}).

\paragraph{Experimental Setup:}
We generate selective explanations and evaluate their MSE and Spearman's correlation to the high-quality explanation computed using a large number of inferences\footnote{We provide details on how high-quality explanations were computed in Appendix \ref{apx:implementation_details}}.
Although our results hold for any feature attribution method, in this section, we focus on Shapley values due to their frequent use and their prevalence in the literature on amortized explainers \citep{jethani2022fastshap, covert2024stochastic, Yang2023EfficientSV}. 
Seaborn \citep{Waskom2021} is used to compute $95\%$ confidence intervals using bootstrap.

\paragraph{Datasets \& Tasks:} We show results for four datasets: two tabular datasets UCI-Adult \citep{misc_adult_2} and UCI-News \citep{misc_online_news_popularity_332}, and two text classification datasets Yelp Review \citep{yelp} and Toxigen \citep{toxigen}.
In the UCI-Adult dataset, the task is to predict if a given individual makes more than $\$50$k a year from a vector with $12$ features; in UCI-News, the task is to predict if a news article will be shared more than $1400$ (median sharing count) times from a vector with $58$ features.
In the Yelp Review dataset, the task is to predict whether a given Yelp review is positive or not, and in the Toxigen dataset, the task is to predict whether a given input text is toxic or not. 
We use $4000$ samples from each dataset due to the computational cost of computing high-quality explanations for this evaluation. \textbf{Models:}
For the tabular datasets, we train a multilayer perceptron \citep{haykin1994neural} to learn the desired task.
We use the HuggingFace Bert-based model \texttt{textattack/bert-base-uncased-yelp-polarity} \citep{morris2020textattack} for the Yelp dataset and the Roberta-based model \texttt{tomh/toxigen\_roberta} \citep{toxigen} for the Toxigen dataset
\footnote{For more details on implementation, please see Appendix \ref{apx:implementation_details}.}.

\paragraph{Efficacy of Uncertainty Measures:}
\label{ssec:coverage_vs_performance}
In Figure \ref{fig:coverage_MSE}, the x-axis shows the coverage ($\alpha$) of the amortized explainer, while the y-axis shows the average mean square error (MSE) \footnote{In Appendix \ref{apx:coverage_vs_performance}, we also show the effect of our uncertainty metrics on Spearman's correlation.} of the selected amortized explanations from high-quality explanations, using deep uncertainty (with 20 models) and learned uncertainty to select which points should fall within the coverage.
The Oracle\footnote{The oracle is computationally expensive because it requires access to high-quality explanations.} is computed by sorting examples from smallest to highest MSE and computing the average MSE for the bottom $\alpha$-fraction of points and is the best that can be done.
Figure \ref{fig:coverage_MSE} shows that both deep uncertainty and learned uncertainty metrics can successfully identify examples that will receive lower and higher-quality explanations.
For the large models ((c) and (d)), the learned uncertainty metric can identify points that will receive low-quality explanations almost as accurately as the Oracle.
Also, we can ensure an MSE smaller than 0.003 (Adult), 0.025 (News), 0.07 (Yelp), and 0.007 (Toxigen) instead of the average MSEs that are 0.014, 0.032, 0.11, and 0.032 respectively with theoretical guarantees \citep{selective_classification} for 50\% of the points as described in Remark \ref{rem:theory_guarantee}. 

\begin{figure}[t]
  \centering
    \begin{subfigure}[b]{0.24\linewidth}
    \includegraphics[width=\linewidth]{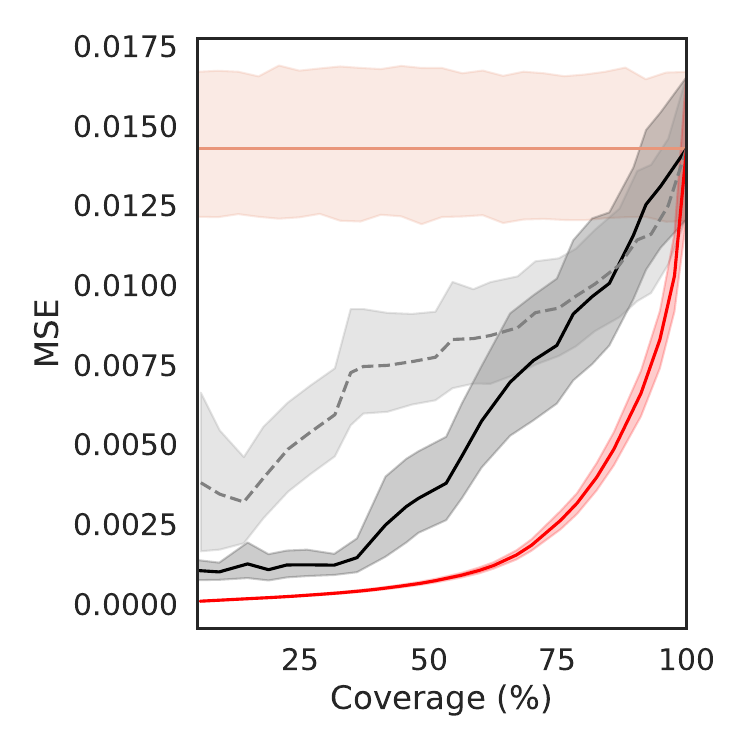}
    \caption{UCI-Adult}
  \end{subfigure}
\begin{subfigure}[b]{0.24\linewidth}
    \includegraphics[width=\linewidth]{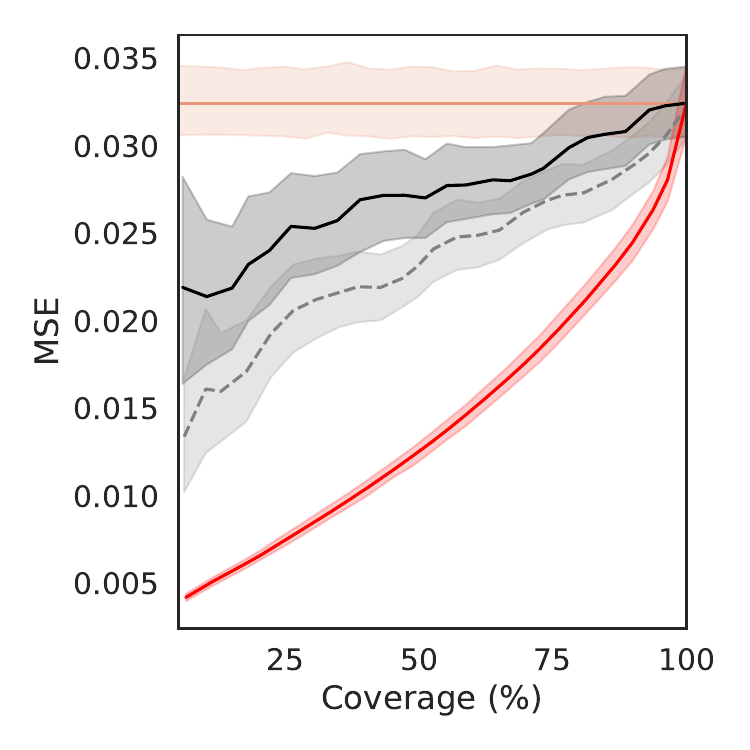}
    \caption{UCI-News}
  \end{subfigure}
  \begin{subfigure}[b]{0.24\linewidth}
    \includegraphics[width=\linewidth]{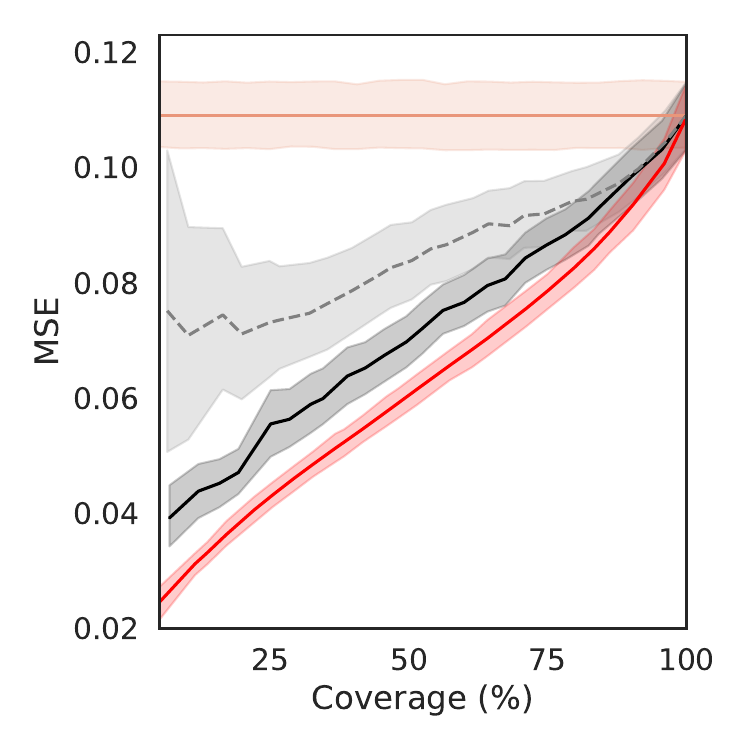}
    \caption{Yelp Review}
  \end{subfigure}
  \begin{subfigure}[b]{0.24\linewidth}
    \includegraphics[width=\linewidth]{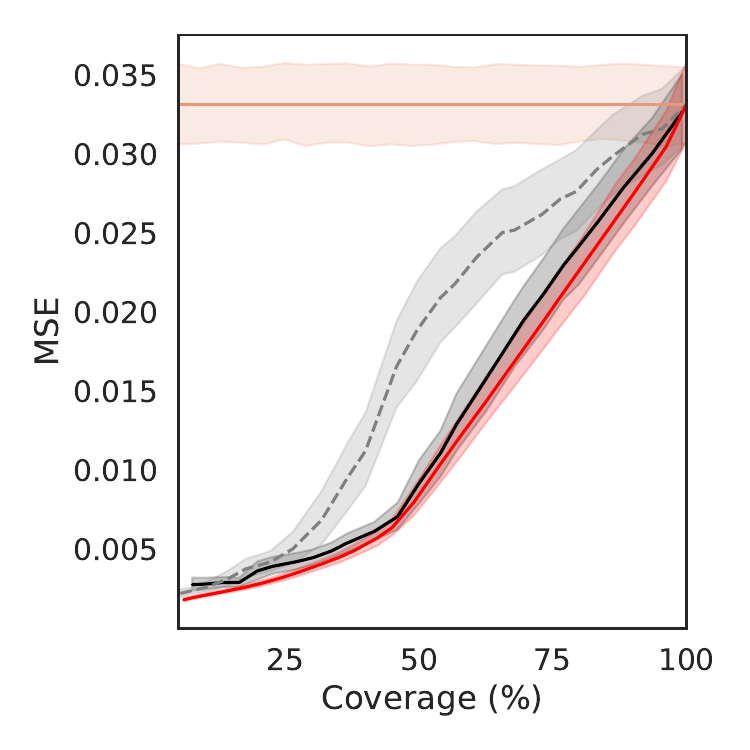}
    \caption{Toxigen}
  \end{subfigure}
  \medskip
    \begin{subfigure}[b]{0.4\linewidth}
    \includegraphics[trim={0 9cm 0 0}, width=\linewidth]
    {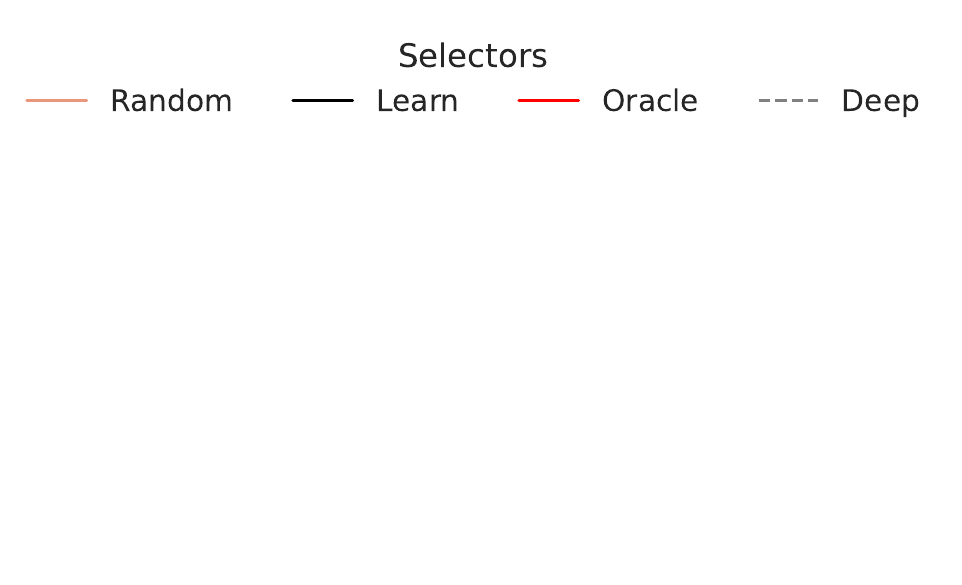}
  \end{subfigure}
  \caption{Coverage ($\alpha$) vs. test MSE from the high-quality explanation.
  The MSE is computed over the points such that $\tau_{\alpha}(\bx) = 1$, i.e., predicted to be high-quality for a given coverage (x-axis).
  When coverage is $100\%$, the MSE is the average performance for the amortized explainer.
  }
  \label{fig:coverage_MSE}
\end{figure}

\begin{figure}[!t]
  \centering
  \medskip
    \begin{subfigure}[b]{0.24\linewidth}
    \includegraphics[width=\linewidth]{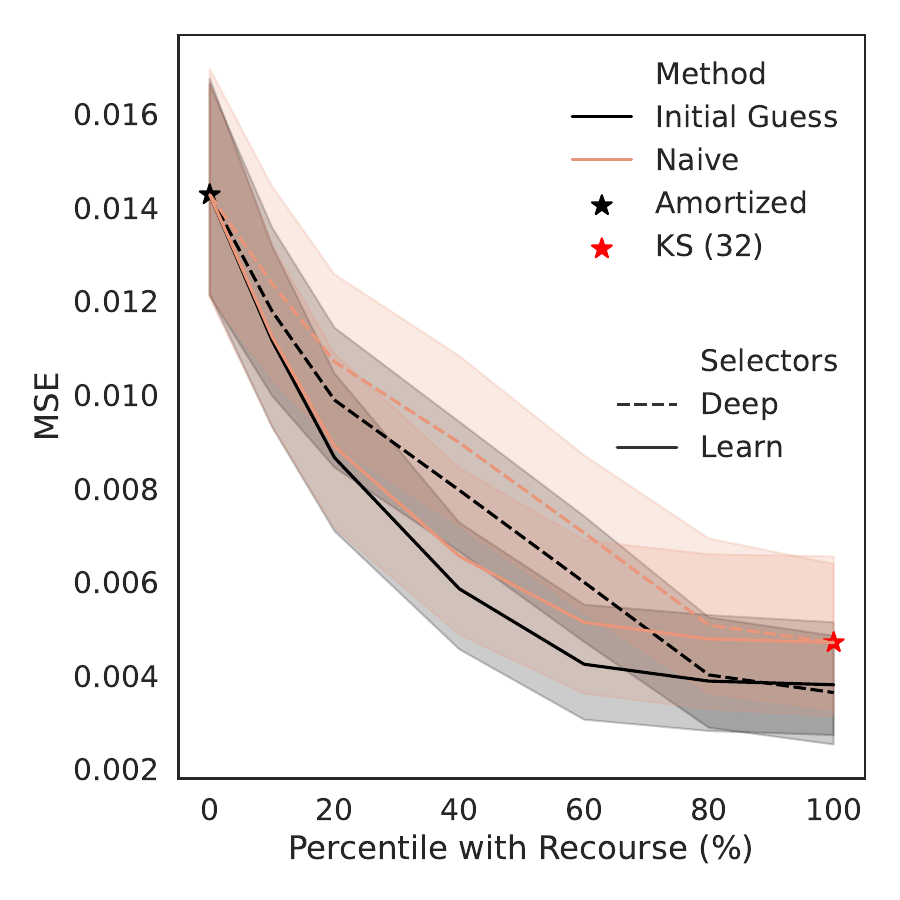}
    \caption{UCI-Adult}
  \end{subfigure}
\begin{subfigure}[b]{0.24\linewidth}
    \includegraphics[width=\linewidth]{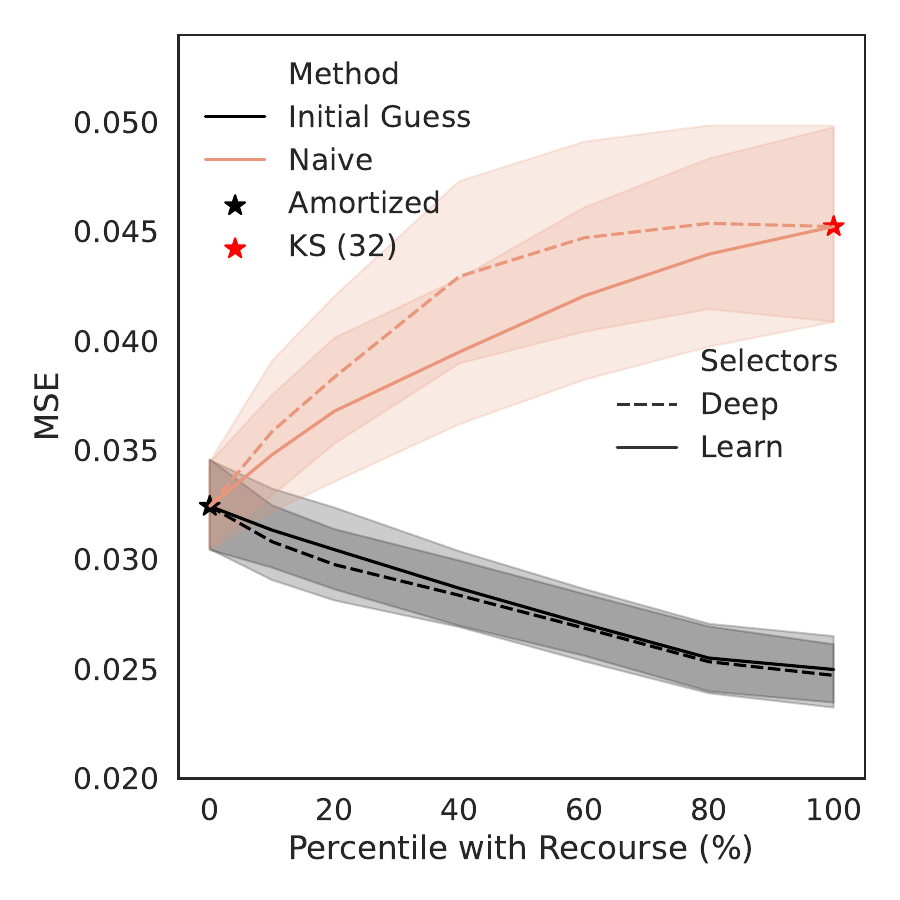}
    \caption{UCI-News}
  \end{subfigure}
  \begin{subfigure}[b]{0.24\linewidth}
    \includegraphics[width=\linewidth]{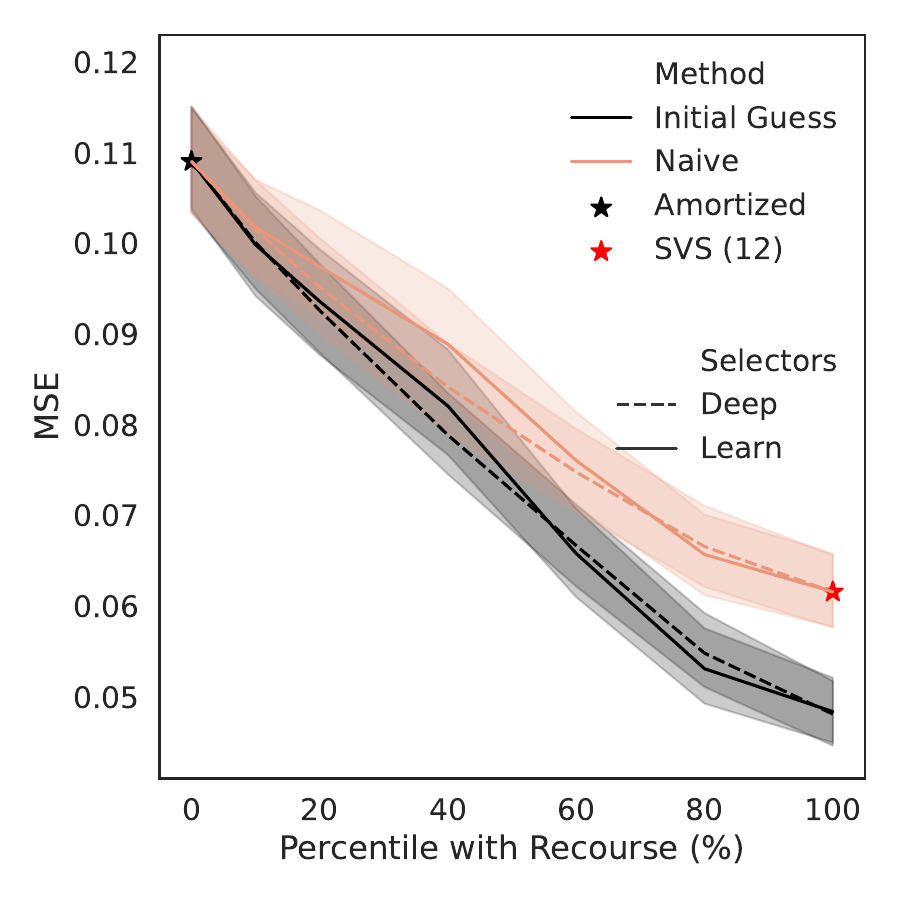}
    \caption{Yelp Review}
  \end{subfigure}
  \begin{subfigure}[b]{0.24\linewidth}
    \includegraphics[width=\linewidth]{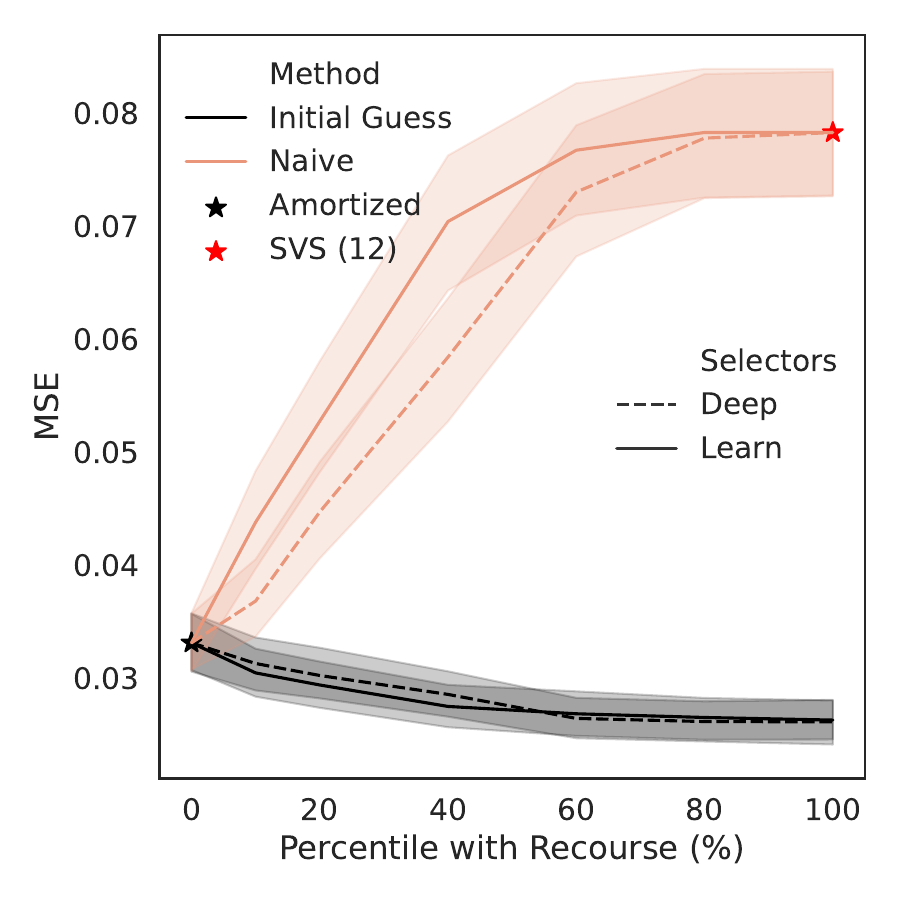}
    \caption{Toxigen}
  \end{subfigure}
  \caption{Fraction ($1-\alpha$) of points that receive explanations with initial guess (x-axis) vs.~MSE of selective explanations w.r.t.~high-quality explanations (y-axis). Naive uses $\lambda_h = 0$ while Initial guess uses $\lambda_h$ in \eqref{eq:lambda_aprox_definition}.
  MSE is computed across all points in the test dataset.
  }
  \label{fig:initial_guess_mse}
\end{figure}

\paragraph{Explanations with Initial Guess:}
\label{ssec:Initial_guess_effect}
In Figure \ref{fig:initial_guess_mse} we compare explanations with initial guess (Definition \ref{def:explanation_initial_guess}) to only using Monte Carlo explanations to provide improve for low-quality explanations, i.e., $\lambda_h = 0$ which we call Naive.
When the MSE from the Monte Carlo is smaller than from the amortized explainer ((a) and (c)), employing explanations with initial guess results in a smaller MSE compared to naively using the Monte Carlo explainer. This suggests that despite their lower quality, the amortized explanations contain valuable information that can be used.
When the Monte Carlo %
is larger than the amortized MSE ((b) and (d)), naive worsens the MSE while explanations with initial guess reduce the MSE, even when using poorer-quality Monte Carlo explanations.
We used KS-32 for the Tabular datasets and SVS-12 for the textual datasets\footnote{In Appendix \ref{apx:Initial_guess_effect}, we also show that selective explanations improve the MSE while maintaining the same level of Spearman's correlation as the Naive approach.}.

\begin{figure}[t]
  \centering
  \begin{subfigure}[b]{0.24\linewidth}
    \includegraphics[width=\linewidth]{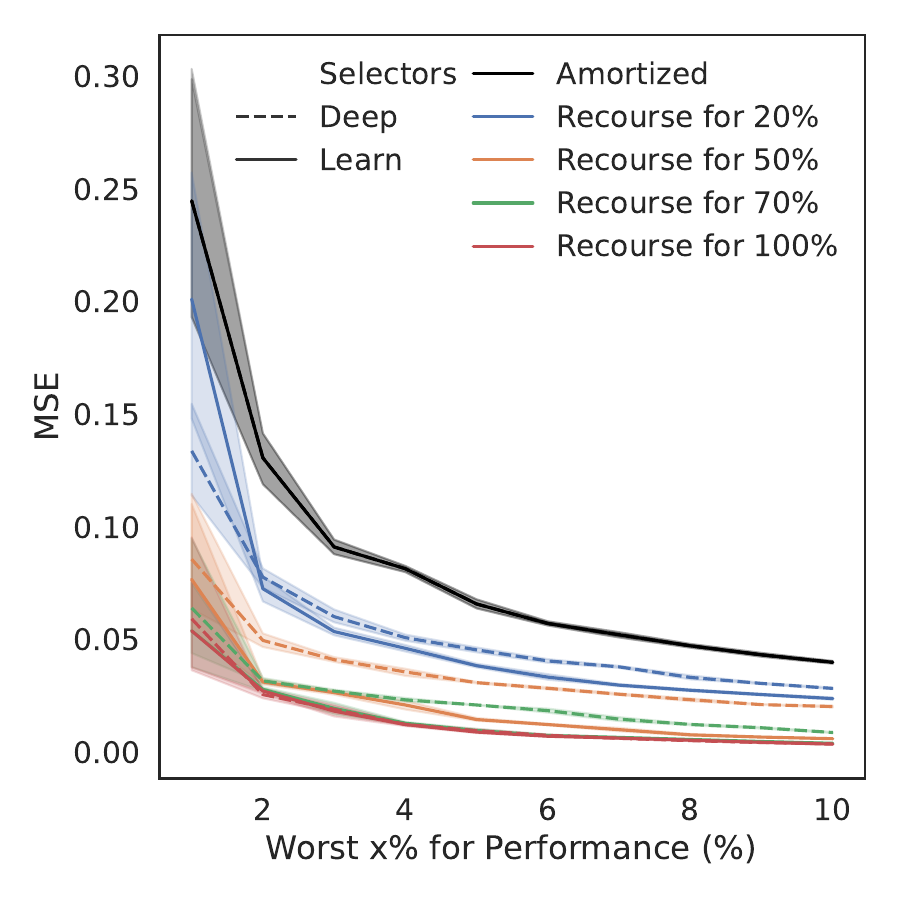}
  \end{subfigure}
  \begin{subfigure}[b]{0.24\linewidth}
    \includegraphics[width=\linewidth]{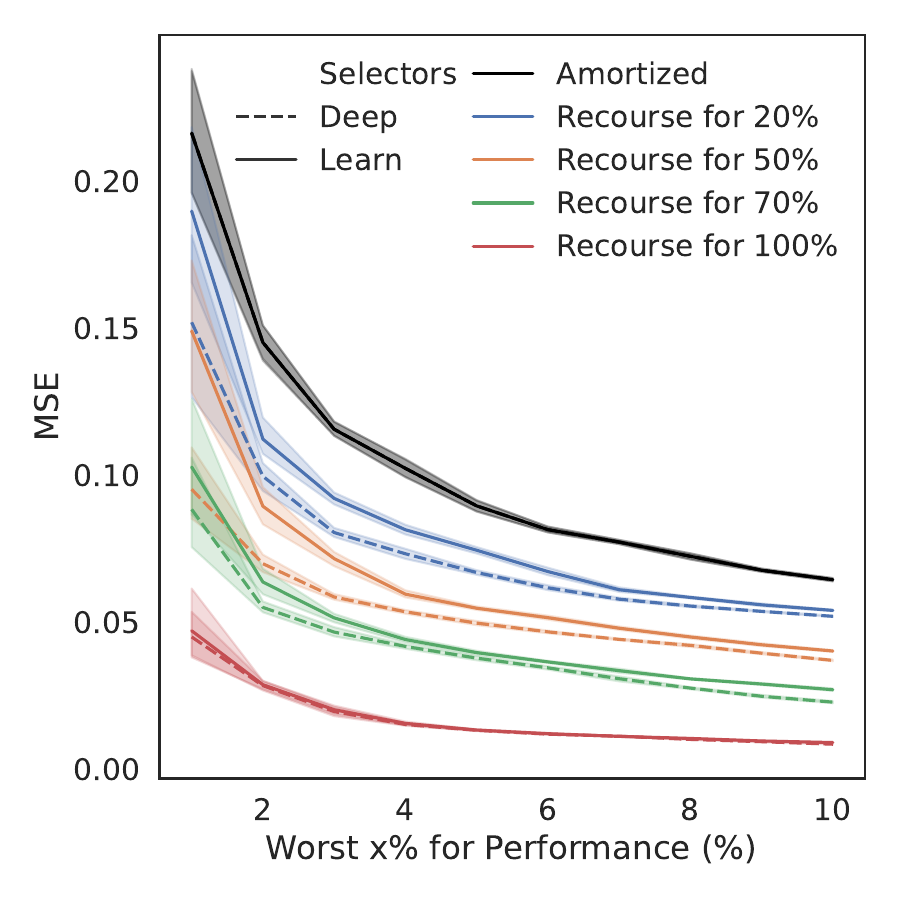}
  \end{subfigure}
  \begin{subfigure}[b]{0.24\linewidth}
    \includegraphics[width=\linewidth]{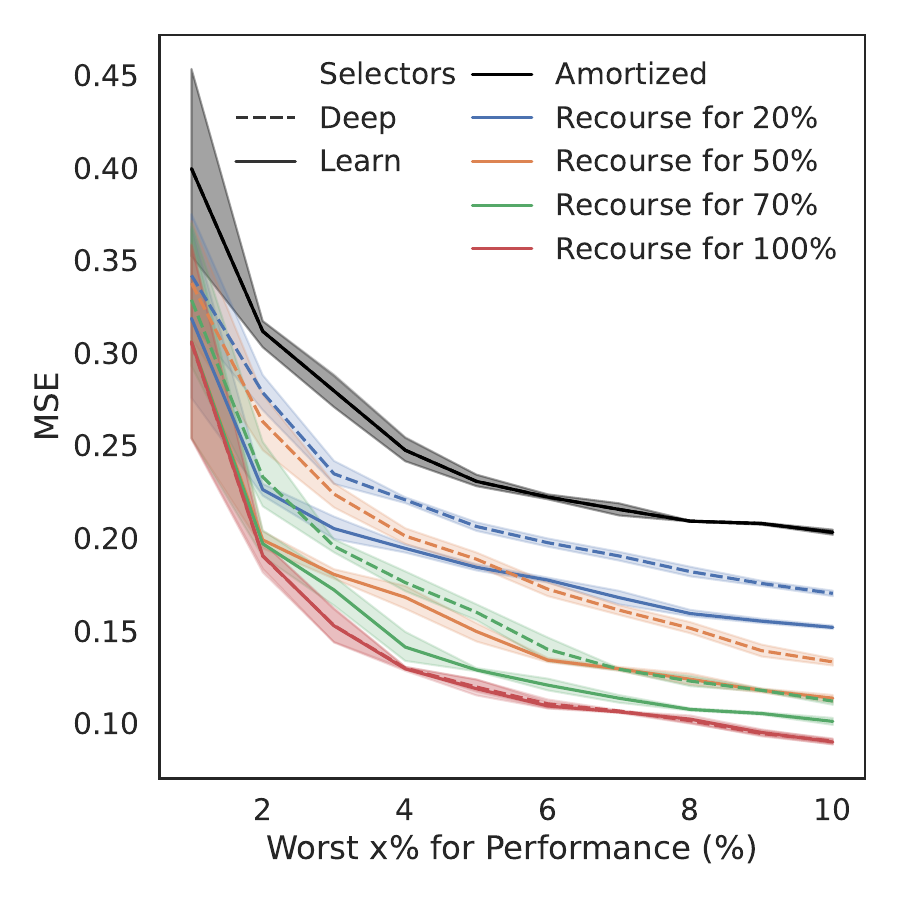}
  \end{subfigure}
  \begin{subfigure}[b]{0.24\linewidth}
    \includegraphics[width=\linewidth]{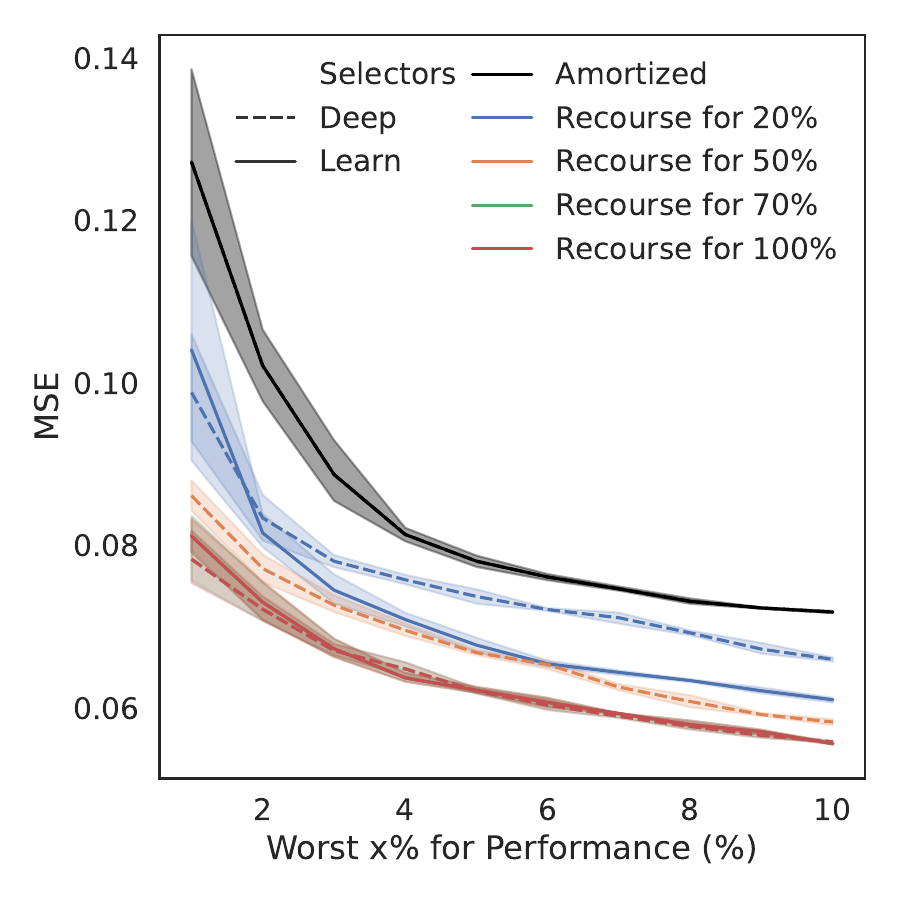}
  \end{subfigure}
  \medskip

    \begin{subfigure}[b]{0.24\linewidth}
    \includegraphics[width=\linewidth]{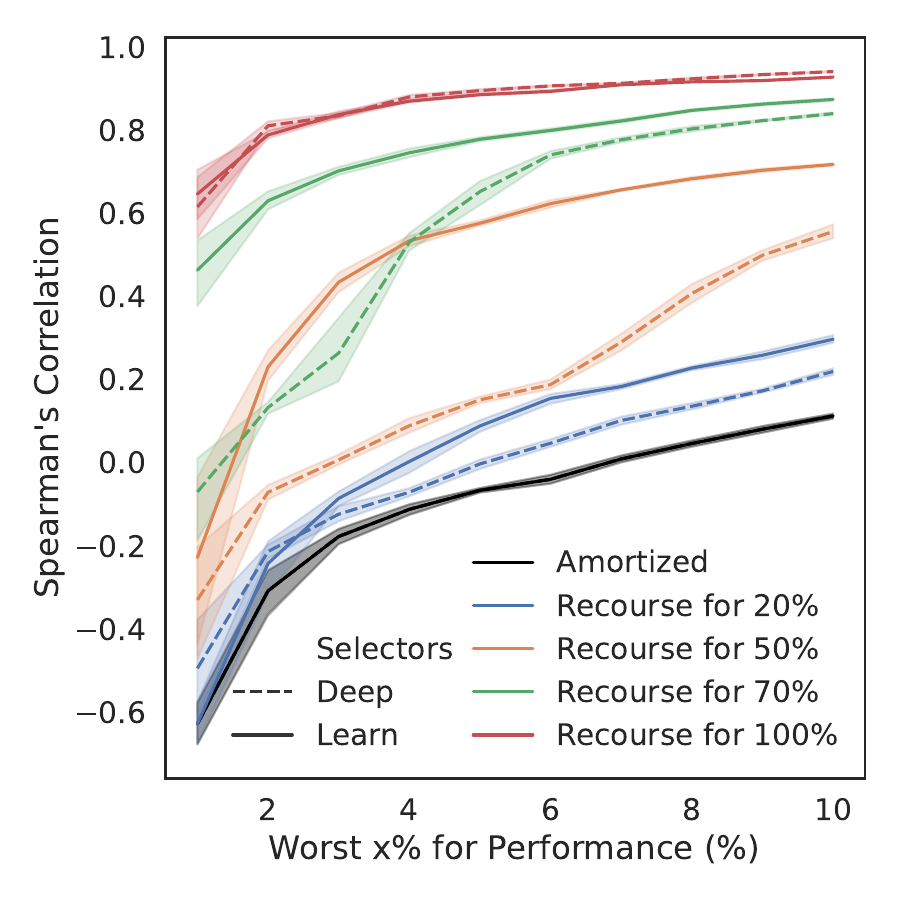}
    \caption{UCI-Adult}
  \end{subfigure}
  \begin{subfigure}[b]{0.24\linewidth}
    \includegraphics[width=\linewidth]{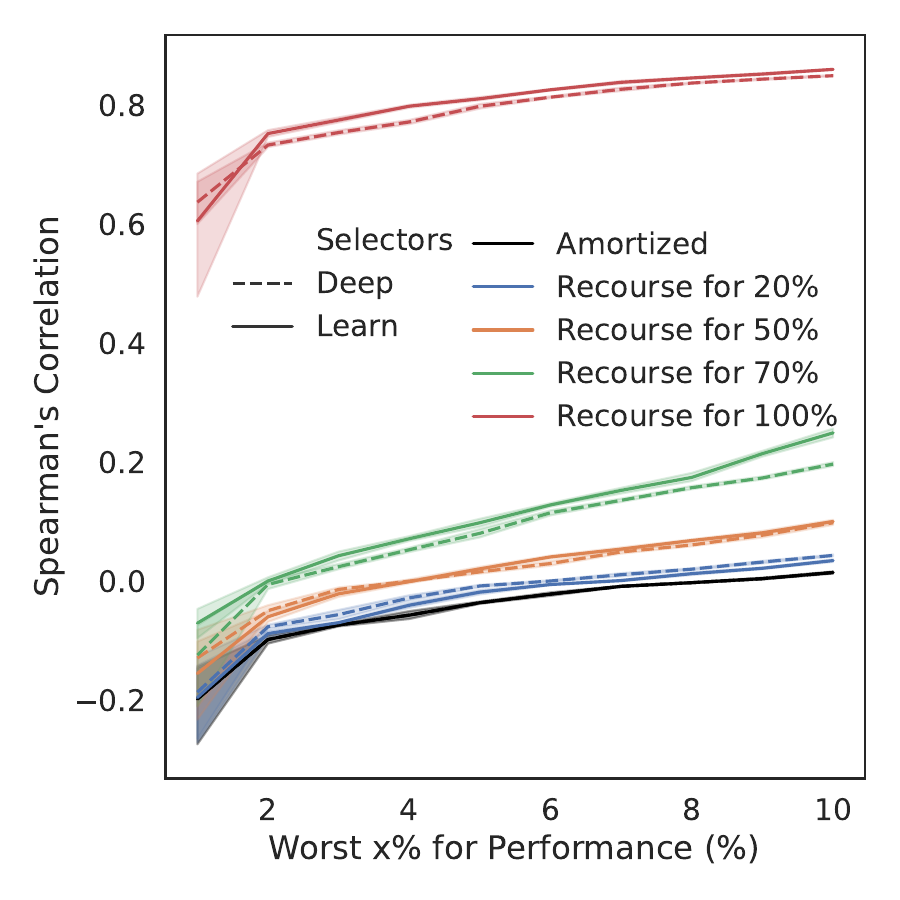}
    \caption{UCI-News}
  \end{subfigure}
  \begin{subfigure}[b]{0.24\linewidth}
    \includegraphics[width=\linewidth]{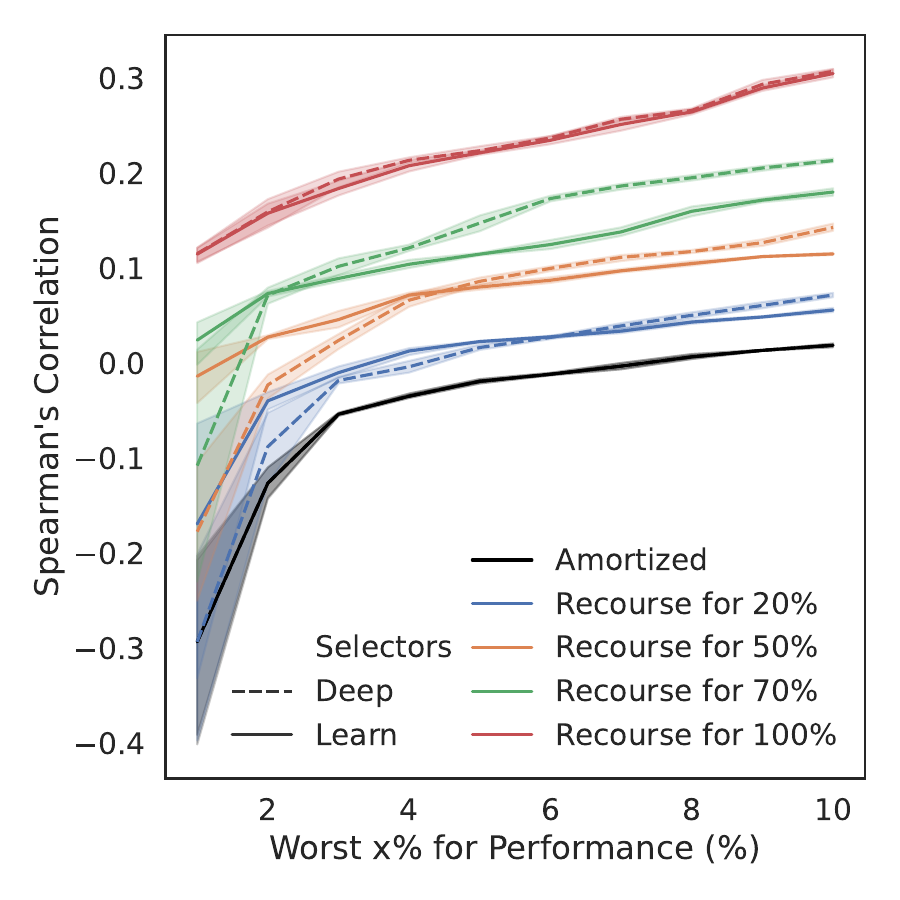}
    \caption{Yelp Review}
  \end{subfigure}
  \begin{subfigure}[b]{0.24\linewidth}
    \includegraphics[width=\linewidth]{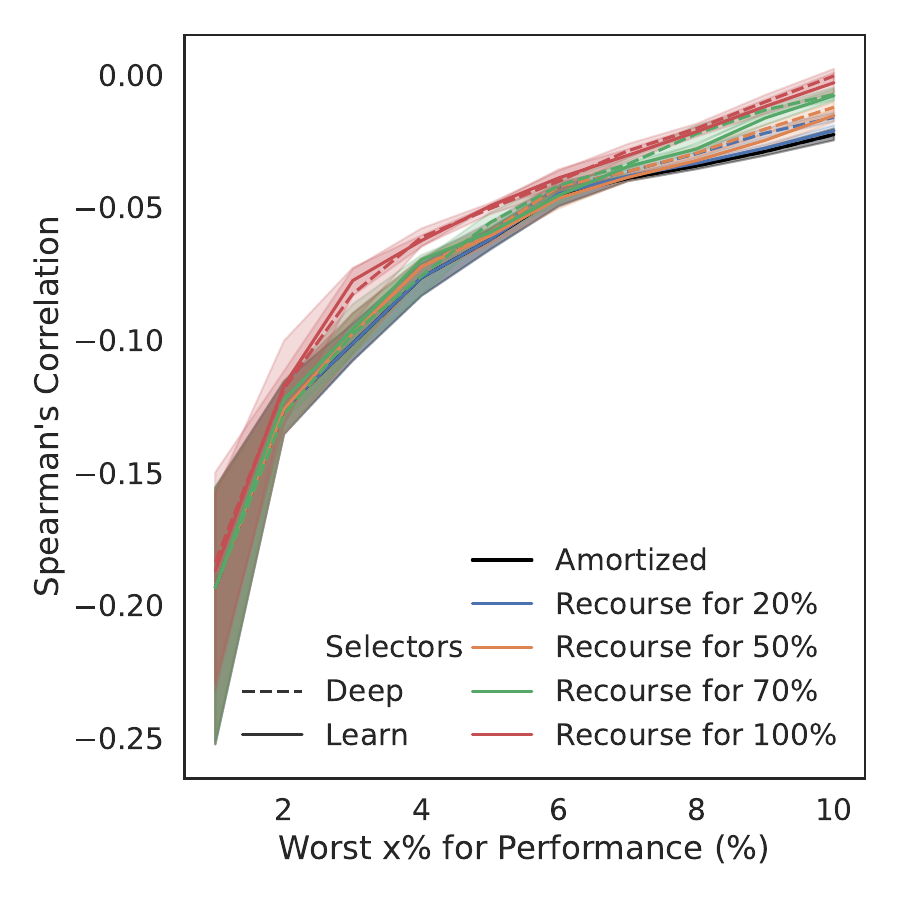}
    \caption{Toxigen}
  \end{subfigure}
  \caption{MSE (top) and Spearman's correlation (bottom) for %
  explanations with the worst performance (highest MSE and smallest Spearman's) in $\calD_{\texttt{test}}$. Colors indicate different percentages $1 - \alpha$ of points receiving explanations with initial guess: $20\%$ (blue), $50\%$ (orange), $70\%$ (green), $100\%$ (red), and amortized explanations $0\%$ (black). Performance is computed in each quantile.}
  \label{fig:performance_in_bottom}
\end{figure}

\paragraph{Worst Case Performance Improvement:}
\label{ssec:improving_worst_case}
In Figure \ref{fig:performance_in_bottom}, we analyze the performance of selective explanations for varying coverages (both in terms of MSE and Spearman's correlation) for the points that receive the worst-performing explanations.
We observe that selective explanations, even with only $20\%$ of points receiving explanations with initial guess, increase Spearman's correlation and decrease MSE consistently across datasets.
Remarkably, when providing explanations with initial guess for $20\%$ of the population in the Yelp dataset (Figure \ref{fig:performance_in_bottom} (c)), selective explanations result in Spearman's correlation for the worst $4\%$ of points that is better than that for the worst $10\%$ from the amortized explainer -- even clearer in the UCI-Adult dataset.
We use SVS-3 for the tabular datasets and  SVS-12 for the text datasets.
\begin{wrapfigure}{r}{8cm}
\centering
\vspace{-0.8cm}
  \begin{subfigure}[b]{0.45\linewidth}
    \includegraphics[width=\linewidth]{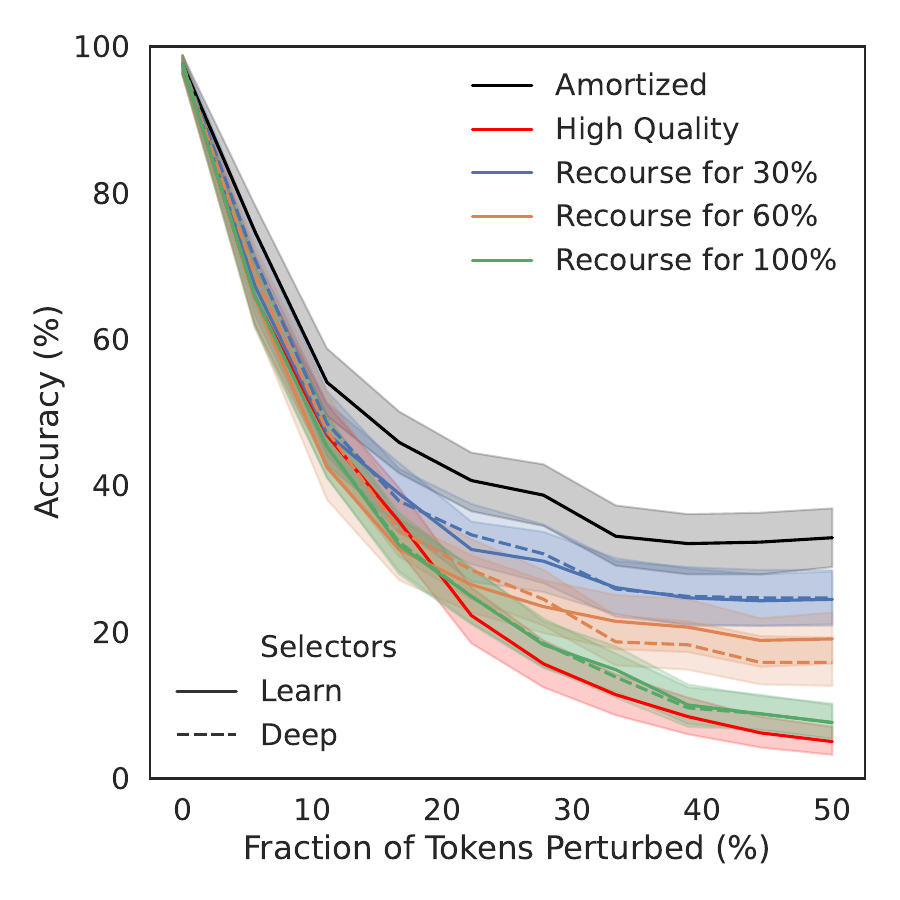}
     \caption{Yelp Review}
  \end{subfigure}
  \begin{subfigure}[b]{0.45\linewidth}
    \includegraphics[width=\linewidth]{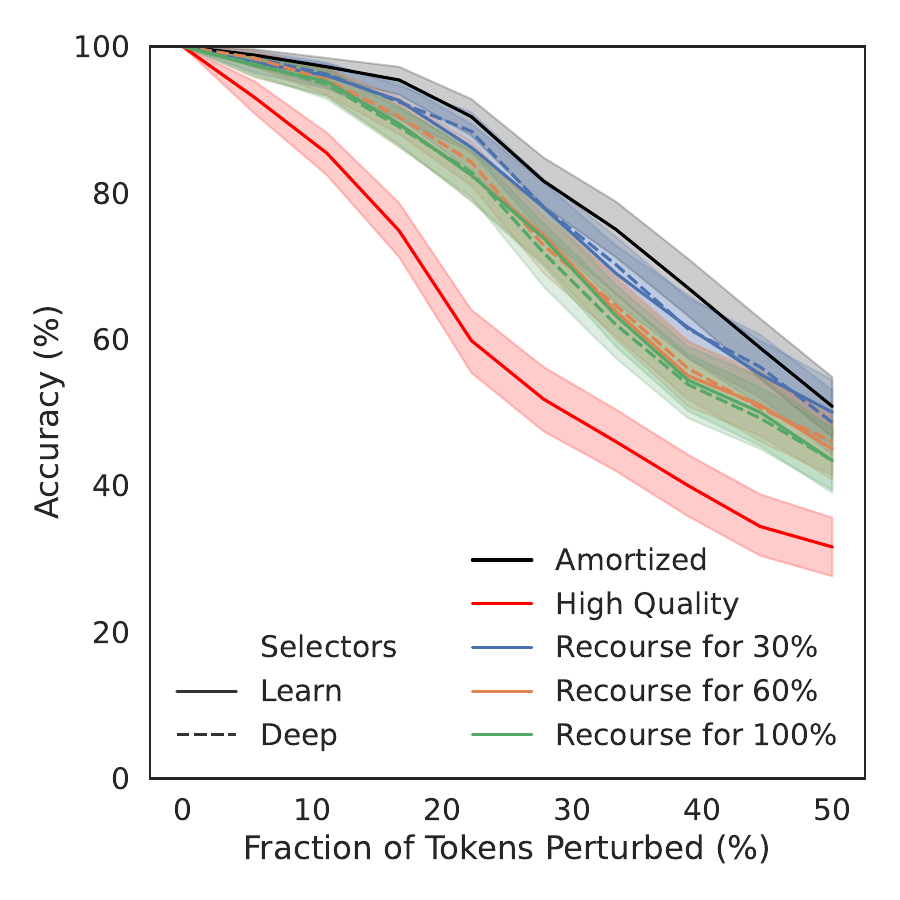 }
     \caption{Toxigen}
  \end{subfigure}
  \caption{Model accuracy (y-axis) when removing the tokens with the highest attribution scores according to the amortized explainer (black), selective explanations with recourse for $30\%$ (blue), $60\%$ (orange), and $100\%$ (green) of points, and high-quality explanations (red).}
    \label{fig:local_fidelity}  
    \vspace{-1.2cm}
\end{wrapfigure}
\paragraph{Perturbation Curve:} 
\label{ssec:perturbation_curve}
Figure \ref{fig:local_fidelity} shows that selective explanations increase the local fidelity of the amortized explainer and that the local fidelity increases with the percentage of points that receive explanations with initial guess (i.e., decreases with coverage).
Both Yelp and Toxigen models are receiving recourse by using SVS-12.
Notably, for Yelp (Figure \ref{fig:local_fidelity} (a)), when providing explanations with initial guess for $60\%$ of the points and using the amortized explainer the other $40\%$ of the time, we achieve local fidelity that is close to the computationally expensive high-quality explanations for the $30\%$ most important tokens.

\section{Final Remarks}
\label{sec:conclusion}

\paragraph{Conclusion:} 
We propose \emph{Selective explanations} that first identify which inputs would receive a low-quality but computationally cheap explanation (amortized) and then perform model inferences to improve the quality of these explanations.
Specifically, we propose \emph{explanations with initial guess} to improve the quality of explanations by combining computationally cheap explanations (amortized) with more expensive explanations (Monte Carlo) using an optimized combination function, improving the explanation performance beyond both explanations.
We perform experiments in large language models and tabular data classifiers empirically demonstrating the efficacy of selective explanations.
Our experiments indicate that selective explanations (i) efficiently identify points that the amortized explainer would produce low-quality explanations, (ii) improve the quality of the worst-quality explanations, and (iii) improve the local fidelity of amortized explanations. 

\textbf{Limitations:} Selective explanations can be applied to any feature attribution method for which amortized and Monte Carlo explainers were developed. However, our empirical results focus on Shapley values. We leave the application of selective explanations to other attribution methods for future work.
Additionally, we focus on large language models (LLMs) used for text classification. Consequently, we do not explore image classifiers, which may also interest the interpretability community.

\section*{Acknowledgements}
The authors thank Amit Dhurandhar for early discussions on the trustworthiness of amortized explainers.
This material is based upon work supported by the National Science Foundation under grants CAREER 1845852, CIF 1900750, CIF 2312667, and FAI 2040880, and awards from Google Research and Amazon.

\bibliographystyle{unsrtnat}
\bibliography{references}

\clearpage
\appendix
\setcounter{theorem}{0}
\setcounter{proposition}{0}

\newpage
\section{Overview}
In this supplementary material we provide the following information:
\begin{itemize}

    \item Appendix \ref{apx:additional)explanations} discuss other high-quality and Monte Carlo explainers.

    \item Appendix \ref{sec:coverage_for_budget} discuss a guide to select the coverage $\alpha$ when the agent providing selective explanations has a budget for the average number of inferences to provide an explanation.
    
    \item Appendix \ref{apx:More_experiments} shows more experimental results on selective explanations.

    \item Appendix \ref{apx:proofs} shows the proofs for the theoretical results in Section \ref{sec:recourse}.
\end{itemize}

\section{Additional Explanation Methods}
\label{apx:additional)explanations}
In this section, we describe high-quality, Monte Carlo, and amortized explainers with further details.

\subsection{High-Quality Explainers}

\textbf{Shapley Values (SHAP)}  \citep{Lundberg_shap}
is a \textbf{high-quality} explainer
that attributes a value $\phi_i$ for each feature $x_i$ in $\bx = (x_1, ..., x_d)$ which is the marginal contribution of feature $x_i$ if the model was to predict $\by$ \eqref{eq:shapley_values}.
\begin{equation}
    \phi_i(\bx, \by) = \frac{1}{d} \sum_{S \subset [d]/\{i\}}  {d - 1 \choose |S|}^{-1} \left( h_{\by}(\bx_{S \cup \{i\}}) - h_{\by}(\bx_{S}) \right).
    \label{eq:shapley_values_apx}
\end{equation}
SHAP has several desirable properties and is widely used. However, as \eqref{eq:shapley_values} indicates, computing Shapley values and the attribution vector $\hq(\bx, \by) = (\phi_1(\bx, \by), ..., \phi_d(\bx, \by))$ requires $2^d$ inferences from $h$, making SHAP impractical for large models where inference is costly.
This has motivated several approximation methods for SHAP, discussed next\footnote{We also discuss Lime and its amortized version in Appendix \ref{apx:additional)explanations}}.

\paragraph{Local Interpretable Explanations (Lime).} Lime is another feature attribution method \citep{tulio_lime} widely used to provide feature attributions. 
It relies on selecting combinations of features, removing these features from the input to generate perturbations, and using these perturbations to approximate the black box model $h$ locally by a linear model.
The coefficients of the linear model are considered to be the attribution of each feature.
Formally, given a weighting kernel $\pi(S)$ and a penalty function $\Omega$, the attribution produced by lime are given by
\begin{equation}
    (\phi, a) = \argmin_{\phi \in \mathbb{R}^{d}, a \in \Reals} \sum_{S \subset [d]} \pi(S) \left( h(\bx_{S}) - a_0 - \sum_{i \in S} \phi_i\right),
    \label{eqapx:Lime}
\end{equation}
where $\hq(\bx, \by) = \phi$.
As in SHAP, to compute the feature attributions using lime, we need to perform a large number of model inferences, which is prohibitive for large models.

\subsection{Monte Carlo Lime}

\textbf{Shapley Value Sampling (SVS)} \citep{Mitchell_svs} is a \textbf{Monte Carlo} explainer that approximates SHAP by restricting the sum in \eqref{eq:shapley_values} to specific permutations of feature. SVS computes the attribution scores by uniformly sampling $m$ features permutations $S_1, ..., S_m$ restricting the sum in \eqref{eq:shapley_values} and performing $n = md +1$ inferences.
We denote SVS that samples $m$ feature permutations by SVS-$m$.

\textbf{Kernel Shap (KS)} \citep{Lundberg_shap} is a \textbf{Monte Carlo} explainer that approximate the Shapley values using the fact that SHAP can be computed by solving the optimization problem
\begin{equation}
    (\phi, a) = \argmin_{\phi \in \Reals^{d}, a \in \Reals} \sum_{i = 1}^{n} \pi(S_i) \left( h(\bx_{S_i}) - a_0 - \sum_{j \in S_i} \phi_j\right),
    \label{eq:Lime_apx}
\end{equation}
using $\pi(S) = {d \choose |S|} |S| (d - |S|)$ and where $\MC^n(\bx, \by) = \phi$.
Kernel Shap samples $n > 0$ feature combinations $S_1, ..., S_n$ and define the feature attributions to be given by the coefficients $\phi$. 
We refer to Kernel Shap using $n$ inferences as KS-$n$. 
We use the KS-$n$ from the Captum library \citep{captum_lib} for our experiments.

\paragraph{ Sample Constrained Lime.} To approximate the attributions from Lime, we consider the sample-contained version of \eqref{eq:Lime_apx}.
Instead of sampling all feature combinations in $[d]$, we only uniformly sample a fixed number $n$ of feature combinations $S_1, ..., S_n$.
For our experiments, shown in the appendix, we use the Sample Constrained Lime from the Captum library \citep{captum_lib}.

\subsection{Amortized Explainers}

\textbf{Stochastic Amortization} \citep{covert2024stochastic} is a \textbf{Amortized} explainer that uses noisy Monte Carlo explanations to learn high-quality explanations. \citet{covert2024stochastic} trained an amortized explainer $\amortized \in \calF$ in a hypothesis class $\calF$ (we use multilayer perceptrons) that takes an input and predicts an explanation.
Specifically, taking the amortized explainer to be the solution of the training problem given in \eqref{eq:training_amortized}.
\begin{equation}
    \amortized \in \argmin_{f \in \calF} \sum_{(\bx, \by) \in \calD_{\text{train}}} \normEuc{ f(\bx, \by) - \MC^{n}(\bx, \by)}^{2}.
    \label{eq:training_amortized_apx}
\end{equation}
We are interested in explaining the predictions of large models for text classification.
However, the approach in \eqref{eq:training_amortized} is only suitable for numerical inputs.
Hence, we follow the approach from \citet{Yang2023EfficientSV} to explain the predictions of large language models, explained next.

\textbf{Amortized Shap for LLMs} \citep{Yang2023EfficientSV}
is a \textbf{Amortized} explainer similar to the one in \eqref{eq:training_amortized} but tailored for LLMs. 
First, the authors note that they can use the LLM to write all input texts $\bx$ as a sequence of token embedding $[e_1(\bx), ..., e_{|\bx|}(\bx)]$ where $e_i(\bx) \in \mathbb{R}^{d}$ denotes the LLM embedding for the $i$-th token contained in the input text $\bx$ and $|\bx|$ is the number of tokens in the input text.
Second, they restrict $\calF$ in \eqref{eq:training_amortized} to be the set of all linear regressions that take the token embeddings  and output the token attribution score.
Then, they solve the optimization problem in
\begin{equation}
    W \in \argmin_{W \in \mathbb{R}^{d}, b \in \mathbb{R}} \sum_{(\bx, \by) \in \calD_{\text{train}}} \sum_{j = 1}^{|\bx|} \normEuc{W^{T}e_j(\bx) + b - \MC^{n}(\bx, \by)_j}^{2},
    \label{eq:training_amortized_llm_apx}
\end{equation}
and define the amortized explainer as
$
    \amortized(\bx) = (W^Te_1(\bx) + b, ..., W^Te_{|\bx|}(\bx) + b).
$

We use stochastic amortization to produce amortized explainers for tabular datasets and Amortized Shap for LLMs to produce explainers for LLM predictions. Both explainers are trained using SVS-12 as $\MC^{n}$.

\section{Selecting Coverage for a Given Inference Budget}
\label{sec:coverage_for_budget}

\paragraph{Determining Coverage from Inference Budget:} Providing explanations with initial guess increases the number of model inferences from 1 when using solely the amortized explainer to $n+1$.
However, a practitioner may have a budget of inferences, i.e., a maximum average number of inferences they are willing to perform to provide an explanation.
We formalize the notion of inference budget in Definition \ref{def:inference_budget}.

\begin{definition}[Inference Budget]
\label{def:inference_budget} 
Denote by $\text{N}(\selective(\bx, \by))$ the number of model inferences to produce the explanation $\selective(\bx, \by)$.
The inference budget $\text{N}_{\texttt{budget}} \in \mathbb{N}$ is the maximum average number of inferences a practitioner is willing to perform per explanation, i.e., it is such that
\begin{equation}
    \text{N}_{\texttt{budget}} \geq \EE{\text{N}(\selective(\bx, \by))}.
    \label{eq:inference_budget}
\end{equation}
\end{definition}

Once an inference budget $\text{N}_{\texttt{budget}}$ is defined, the coverage $\alpha$ should be set to follow it.
In Proposition \ref{prop:coverage_for_budget}, we show the minimum coverage for the selective explanations to follow the inference budget. 
\begin{proposition}[Coverage for Inference Budget]
\label{prop:coverage_for_budget}
Let $\text{N}_{\texttt{budget}} \geq 1$ be the inference budget, and assume that the Monte Carlo method $\MC^n(\bx, \by)$ uses $n$ model inferences.
Then, the coverage level $\alpha$ should be chosen such that
\begin{equation}
    \frac{n + 1 - \text{N}_{\texttt{budget}}}{n} = \min_{\alpha \in [0, 1]} \alpha, \text{  such that } \EE{\text{N}(\selective(\bx, \by))} \leq \text{N}_{\texttt{budget}} .
\end{equation}
Recall that SVS-$m$ performs $n = 1 + dm$ inferences ($\bx \in \mathbb{R}^d$), and KS-$m$ performs $n = m$ inferences.
\end{proposition}

\section{More Experimental Results}
\label{apx:More_experiments}
In this section, we (i) give further implementation details and (ii) discuss further empirical results. 

\subsection{More Details on Experimental Setup}
\label{apx:implementation_details}

\paragraph{High-Quality Explanations:} We define the high-quality explanations for the tabular datasets to be given by Kernel Shap with as many inferences as needed for convergence, using the Shapley Regression library \citep{Covert2020ImprovingKP}.
For the textual dataset, following \citep{Yang2023EfficientSV}, we define the high-quality explanations to be given by Kernel Shap using $8912$ model inferences per explanation.

\paragraph{Amortized Explainers:} For the tabular datasets, we use the amortized explainer from \citep{covert2024stochastic} that we describe in Section \ref{sec:background}.
Specifically, we use a multilayer perceptrom model architecture to learn the shapley values for the tabular datasets.
For the textual datasets, we use the linear regression on token-level textual embeddings to learn the shapley values, as described in Section \ref{sec:background}.
Both amortized models learn from the training dataset of explanations generated using Shapley Value Sampling from the Captum library \citep{captum_lib} with parameter $12$, i.e., SVS-12.

\paragraph{Uncertainty Metrics:} We test the two proposed uncertainty metrics in Section \ref{sec:selecting}, namely, deep uncertainty and uncertainty learn.
For \textbf{deep uncertainty}, we run the training pipeline for the amortized explainers 20 times for each dataset we perform experiments on, resulting in 20 different amortized explainer that we use to compute \eqref{eq:DeepUncertainty}.
For \textbf{uncertainty learn}, we use the multilayer perceptrom as the hypothesis class with only one hidden layer. The hidden layer was composed of $\kappa = 3d$ neurons where $d$ is the dimension of the input vector $\bx \in \mathbb{R}^d$.
The uncertainty learn metric was trained on $\calD_{\texttt{train}}$, the same training dataset as the amortized explainers.

\paragraph{Dataset sizes:} We use 4000 samples from each dataset due to computational limitations on the computation of high-quality explanations used to evaluate selective explanations. All explanations were computed using the Captum library \citep{captum_lib}.
The dataset $\calD$ with $N = 4000$ samples was partitioned in three parts, $\calD_{\texttt{train}}$ with $50\%$ of points, $\calD_{\texttt{cal}}$ with $25\%$ of points, and $\calD_{\texttt{test}}$ with the other $25\%$ of points.

\paragraph{Computational Resources:} All experiments were run in a A100 40 GB GPU. For each dataset, we compute different Monte Carlo explanations. For the UCI-News dataset, the high quality explanations took 4:30 hours to be generate until convergence while for UCI-Adult it took 3:46 hours. For the tabular datasets, all other Monte Carlo explainers were generated in less than 1 hour.
For the language models, the high-quality explanations with 8192 model inferences, took 18:51 hours for the Toxigen dataset and 20:00 hours for the Yelp Review datasets.
The other used Monte Carlo explanations took proportional (to the number of inferences) time to be generated.

\subsection{Uncertainty Measures Impact on Spearman's Correlation}
\label{apx:coverage_vs_performance}

Figure \ref{fig:coverage_sp} shows in the x-axis the coverage ($\alpha$) and in the y-axis the average Spearman's correlation of the selected amortized explanations from high-quality explanations using deep uncertainty (with 20 models) and the uncertainty learn to select low-quality explanations.
The Oracle\footnote{The oracle is computationally expensive because it requires access to high-quality explanations.} is computed by sorting examples by the smallest to higher MSE and computing the average Spearman's correlation in the bottom x-axis points accordingly to the MSE and is the best that can be done in terms of MSE.

Figure \ref{fig:coverage_sp} shows that the Oracle and proposed uncertainty metrics don't always select the points with the smallest Spearman's correlation first.
This implies that MSE and Spearman's correlation don't always align, i.e., there are points with high MSE and high Spearman's correlation at the same time.
However, we note that the uncertainty learns selector can be applied to \textbf{any} metric $\ell$ as we define in \eqref{eq:UncLearn} including Spearman's correlation and any combination of Spearman's correlation and MSE aiming to approximate both metrics.
Moreover, when the smallest MSE aligns with the highest Spearman's correlation, i.e., the oracle is decreasing in Spearman's correlation when the coverage increases (Figure \ref{fig:coverage_sp} (a) and (c)), the proposed uncertainty metrics also accurately detect the low-quality explanations in term of Spearman's correlation.

\begin{figure}[htb]
  \centering
    \begin{subfigure}[b]{0.24\linewidth}
    \includegraphics[width=\linewidth]{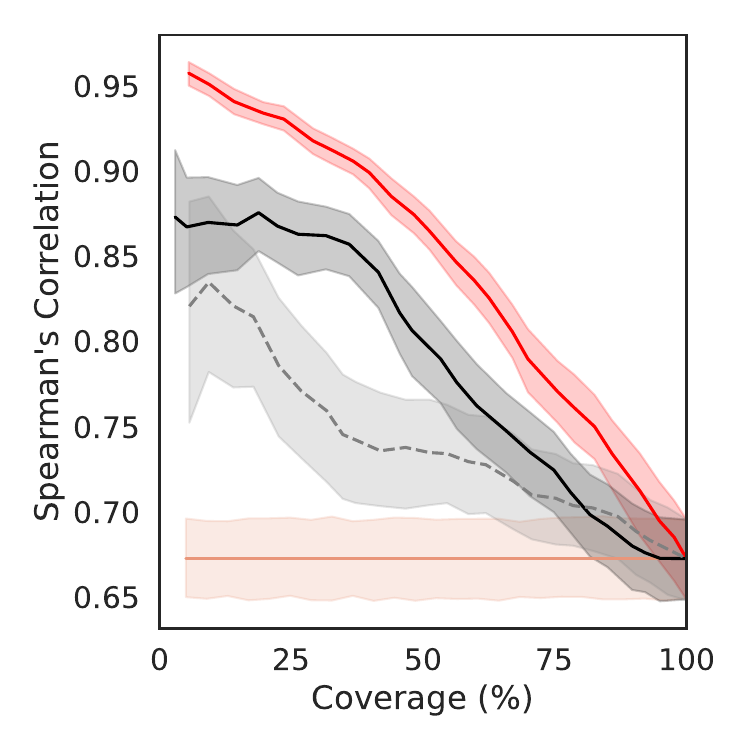}
    \caption{UCI-Adult}
  \end{subfigure}
\begin{subfigure}[b]{0.24\linewidth}
    \includegraphics[width=\linewidth]{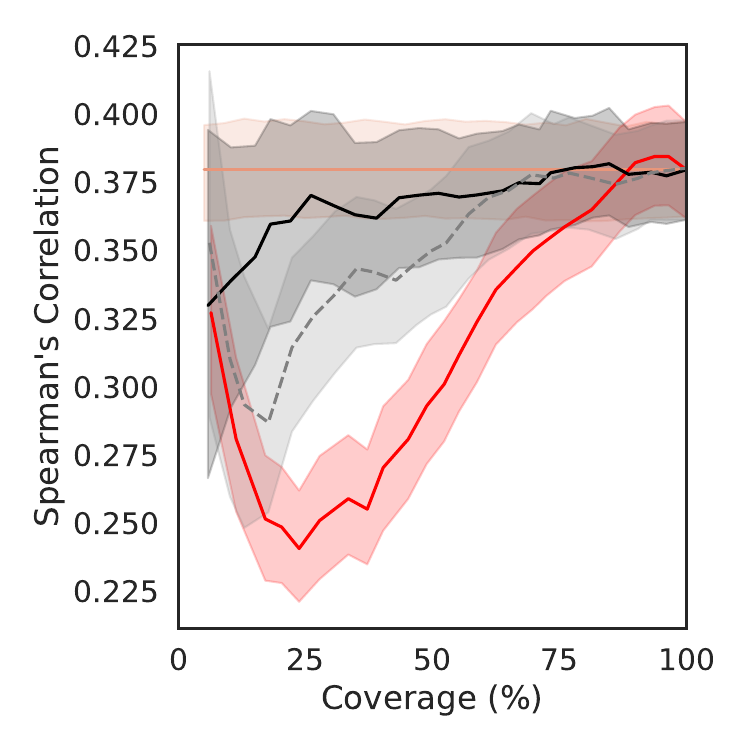}
    \caption{UCI-News}
  \end{subfigure}
  \begin{subfigure}[b]{0.24\linewidth}
    \includegraphics[width=\linewidth]{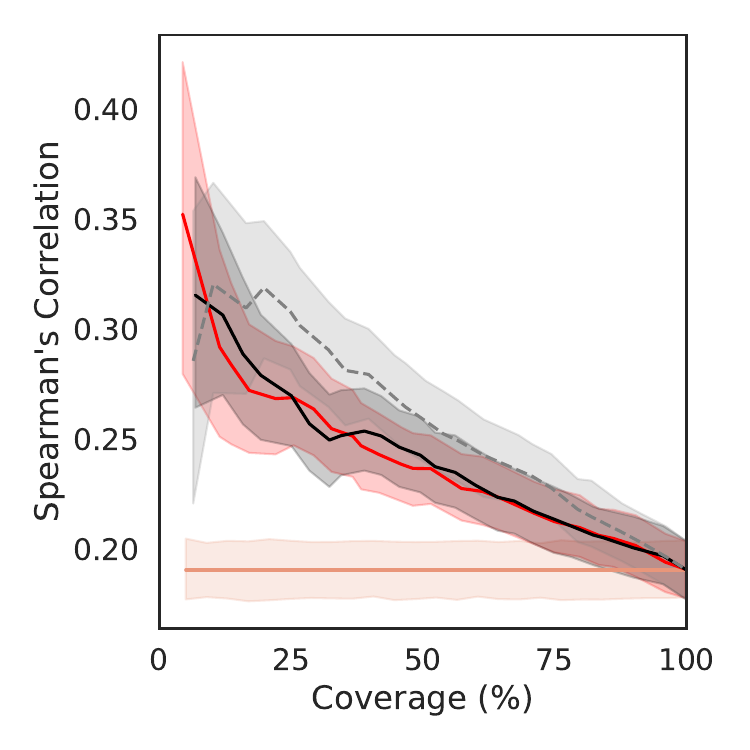}
    \caption{Yelp Review}
  \end{subfigure}
  \begin{subfigure}[b]{0.24\linewidth}
    \includegraphics[width=\linewidth]{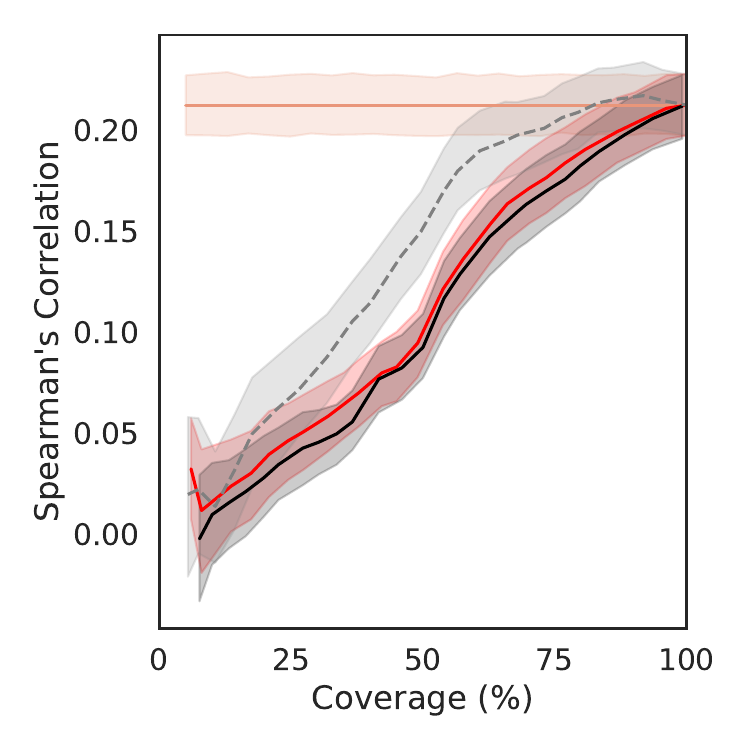}
    \caption{Toxigen}
  \end{subfigure}
  \medskip
    \begin{subfigure}[b]{0.4\linewidth}
    \includegraphics[trim={0 9cm 0 0}, width=\linewidth]
    {CoverageMSE/legend.pdf}
  \end{subfigure}
  \caption{Coverage vs. Spearman's correlation from the high-quality explanation. Coverage is the percentage of the points that the selection function predicts that will receive a higher-quality explanation, i.e., $\tau_t(\bx) = 1$. When coverage is $100\%$ Spearman's correlation is the average performance for the amortized explainer.}
  \label{fig:coverage_sp}
\end{figure}

\subsection{The Effect of Explanations with Initial Guess}
\label{apx:Initial_guess_effect}

In Figure \ref{fig:initial_guess_sp} we compare explanations with initial guess (Definition \ref{def:explanation_initial_guess}) to only using the Monte Carlo to provide recourse to the low-quality explanaitons, i.e., $\lambda_h = 0$ we call it Naive.
In all tested cases, Spearman's correlation of the Monte Carlo method is comparable to or larger than the amortized explainer.
Although selective explanations optimized for MSE by using explanations with initial guess (Definition \ref{def:explanation_initial_guess}), we observe that the Spearman's correlation of selective explanations is close to or larger than the naive method, once again, demonstrating the efficacy of selective explanations.

\begin{figure}[htb]
  \centering
    \begin{subfigure}[b]{0.24\linewidth}
    \includegraphics[width=\linewidth]{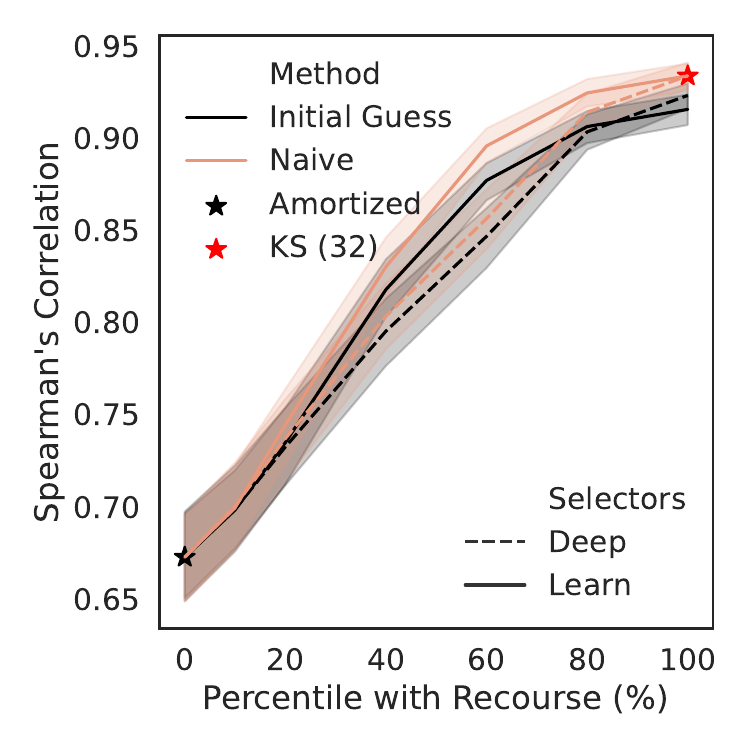}
    \caption{UCI-Adult}
  \end{subfigure}
\begin{subfigure}[b]{0.24\linewidth}
    \includegraphics[width=\linewidth]{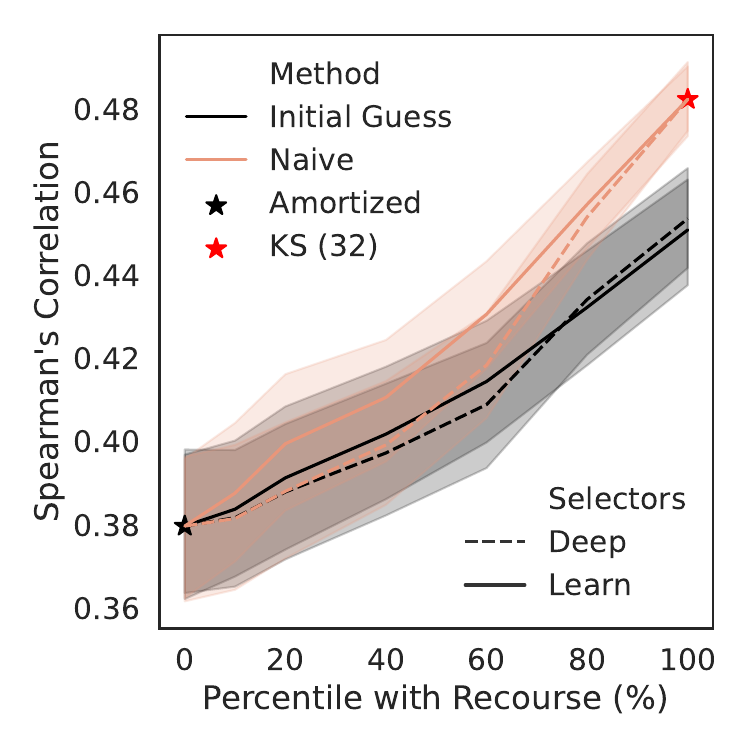}
    \caption{UCI-News}
  \end{subfigure}
  \begin{subfigure}[b]{0.24\linewidth}
    \includegraphics[width=\linewidth]{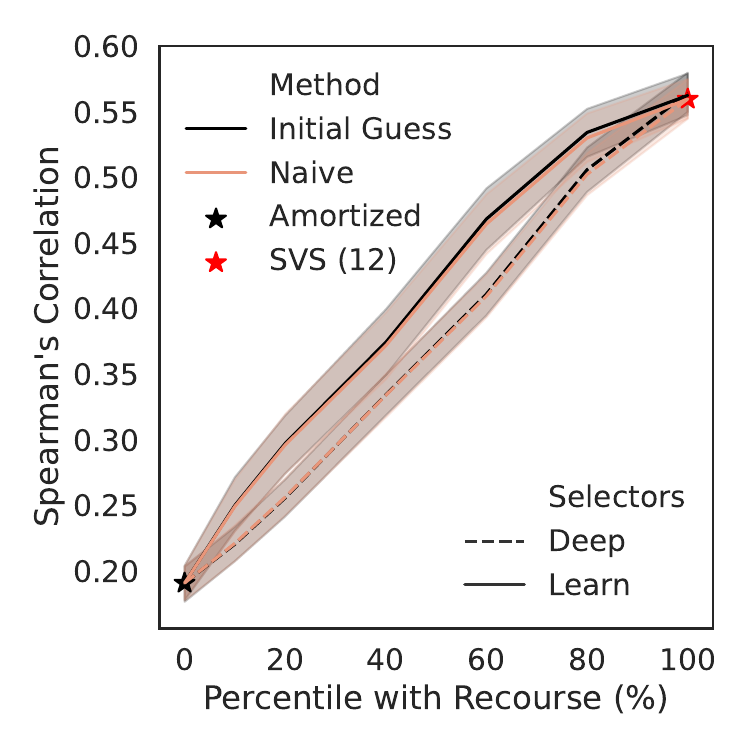}
    \caption{Yelp Review}
  \end{subfigure}
  \begin{subfigure}[b]{0.24\linewidth}
    \includegraphics[width=\linewidth]{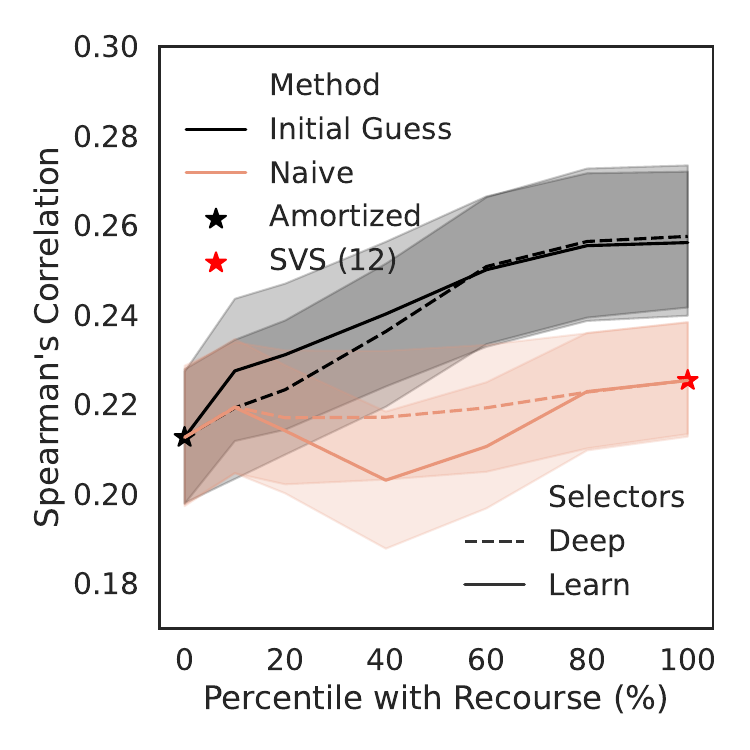}
    \caption{Toxigen}
  \end{subfigure}
  \caption{Fraction of the population that receive explanations with initial guess (x-axis) vs. their Spearman's correlation from the high-quality explanations (y-axis). Naive uses $\lambda_h = 0$ while initial guess uses explanations with initial guess, i.e., when $\lambda_h$ is given in \eqref{eq:lambda_aprox_definition}.}
  \label{fig:initial_guess_sp}
\end{figure}

\subsection{Performance for Different Monte-Carlo Explainers}
\label{apx:ablation_different_recourse}
Figure \ref{fig:various_methods} shows how the MSE and Spearman's correlation behave accordingly with the quality of the Monte Carlo explainer. 
We compare Kernel Shap and Shapley Value Sampling in all experiments.
We observe that when the quality of the Monte Carlo explainer increases, the quality of the Selective explanation also increases, i.e., the MSE decreases and the Spearman's correlation increases.
Moreover, we also observe diminishing returns, i.e., after a certain point, increasing the quality of the Monte Carlo explanations doesn't lead to a tailored increase in performance.
For example, observe the SVS method in the tabular datasets Figure \ref{fig:various_methods} (a) and (b).
We also observe that providing explanations with initial guess has a high impact on both Spearman's correlation and MSE when only providing recourse toa small fraction of the population.
For example, when providing explanations with initial guess for $20\%$ of the population using SVS-12 in the Yelp Review dataset, Figure \ref{fig:various_methods} (c), increases the Spearman's correlation in more than $50\%$ (from 0.2 to more than 0.3). 

\begin{figure}[t]
  \centering
  \begin{subfigure}[b]{0.24\linewidth}
    \includegraphics[width=\linewidth]{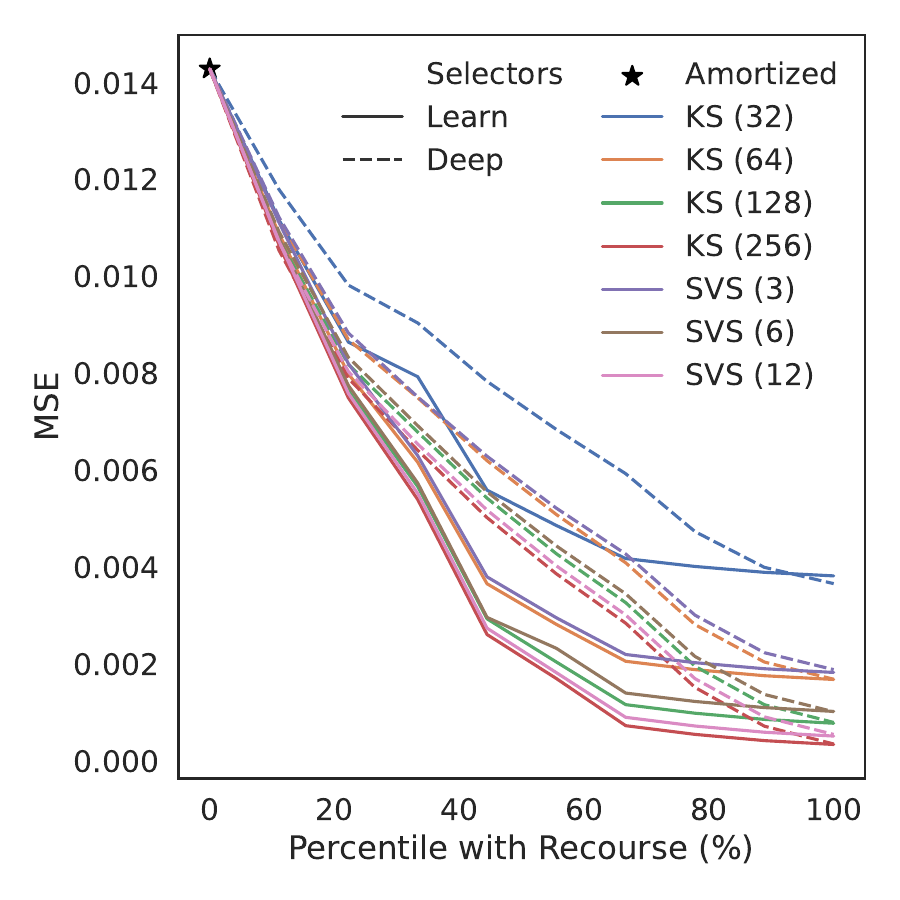}
  \end{subfigure}
  \begin{subfigure}[b]{0.24\linewidth}
    \includegraphics[width=\linewidth]{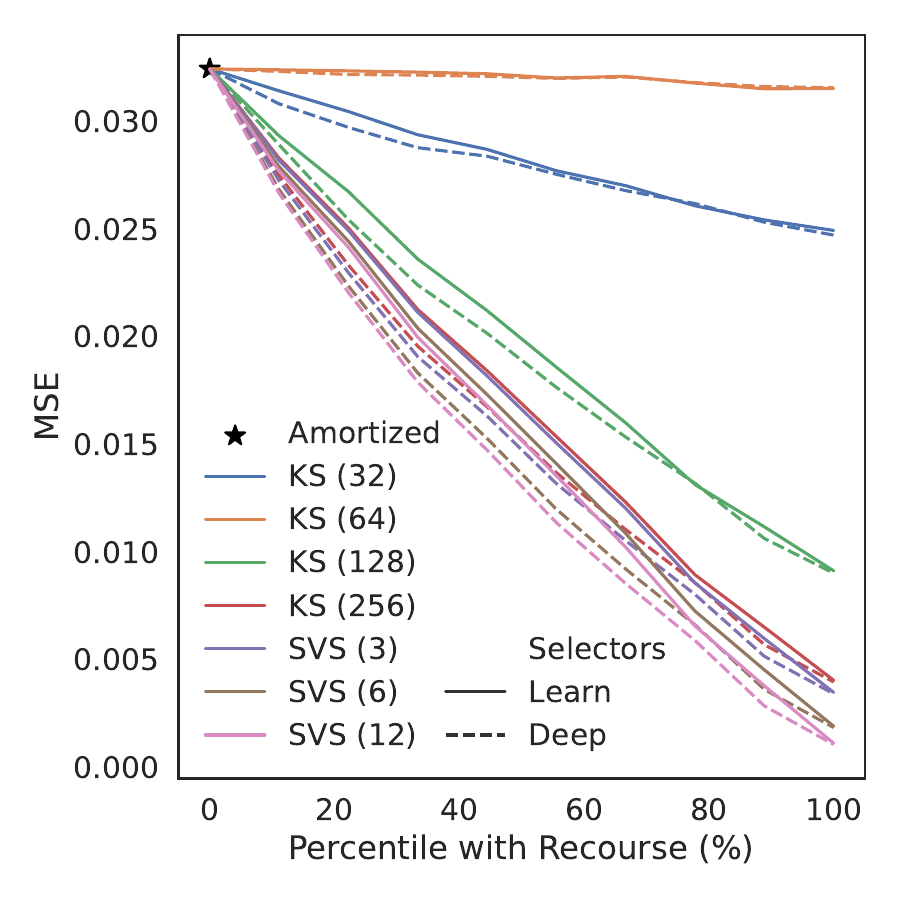}
  \end{subfigure}
  \begin{subfigure}[b]{0.24\linewidth}
    \includegraphics[width=\linewidth]{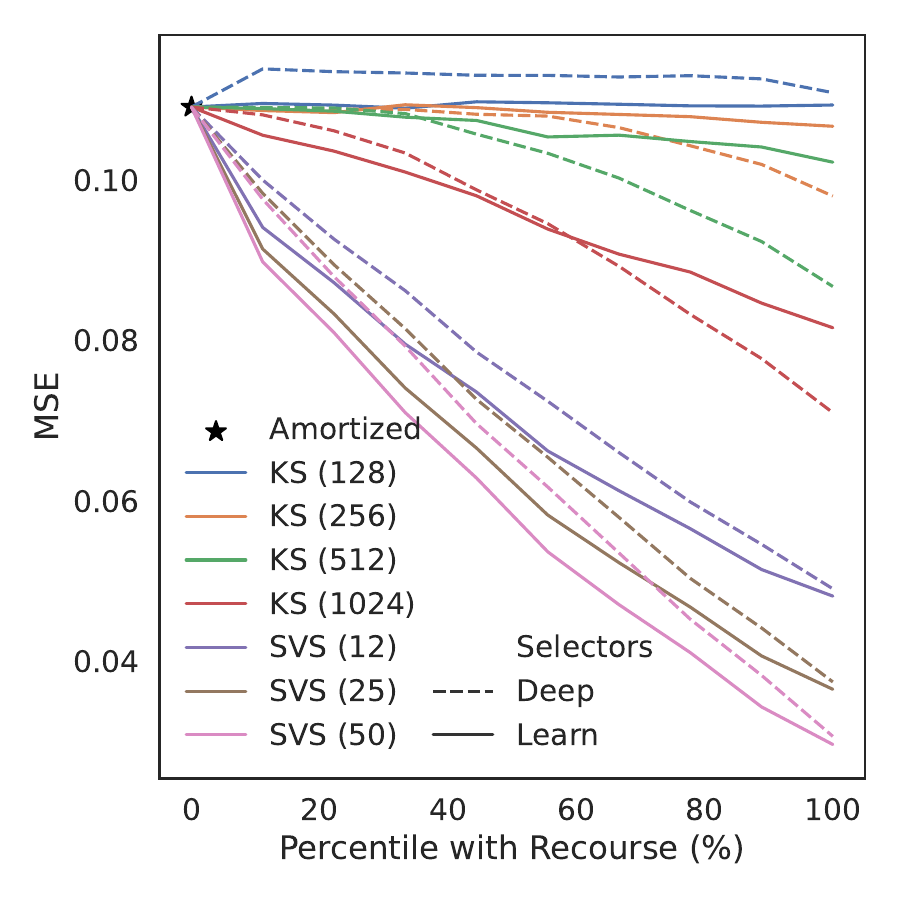}
  \end{subfigure}
  \begin{subfigure}[b]{0.24\linewidth}
    \includegraphics[width=\linewidth]{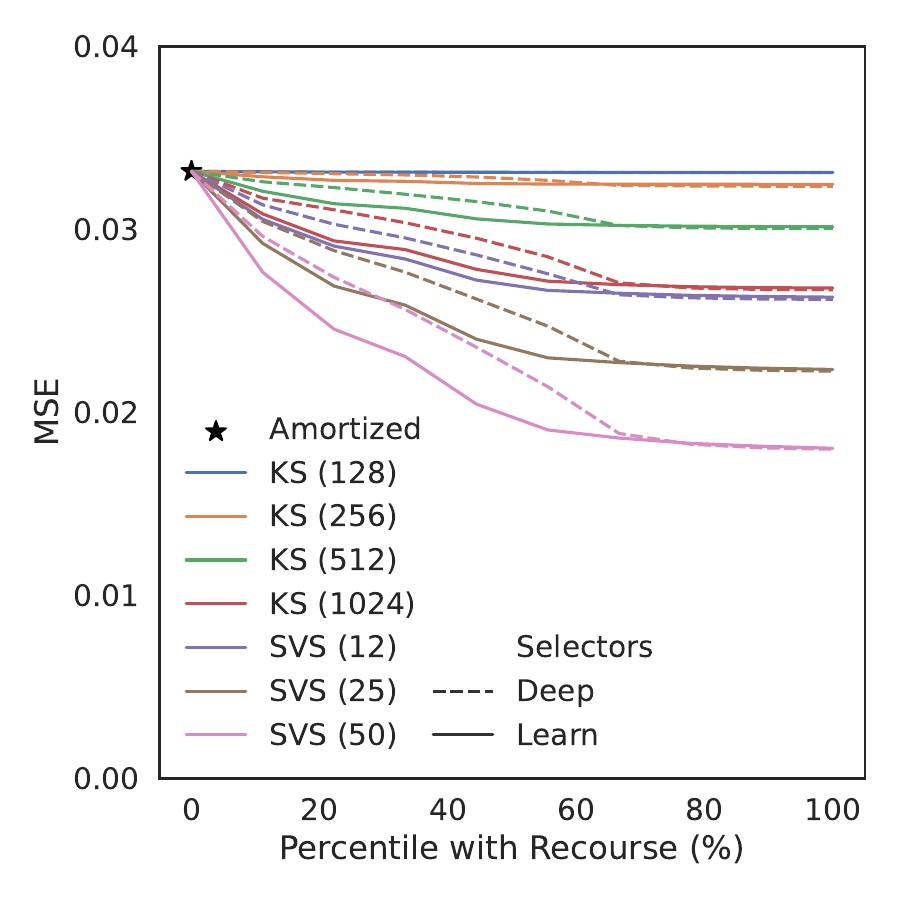}
  \end{subfigure}
  \medskip

    \begin{subfigure}[b]{0.24\linewidth}
    \includegraphics[width=\linewidth]{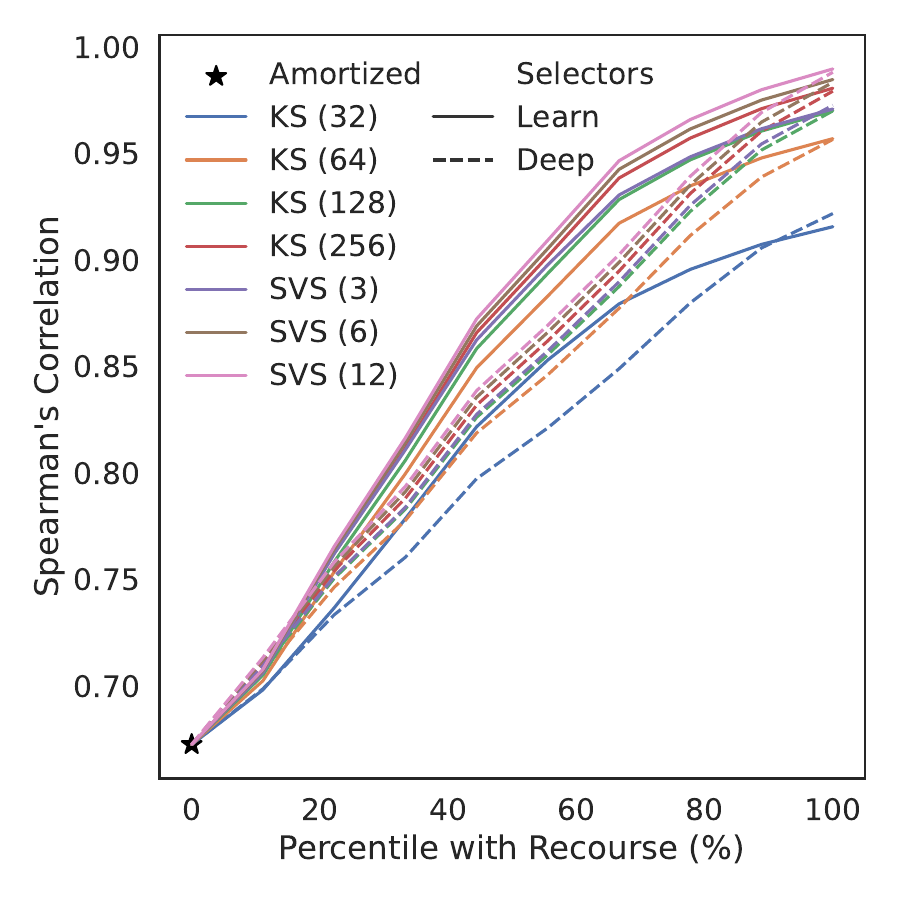}
    \caption{UCI-Adult}
  \end{subfigure}
  \begin{subfigure}[b]{0.24\linewidth}
    \includegraphics[width=\linewidth]{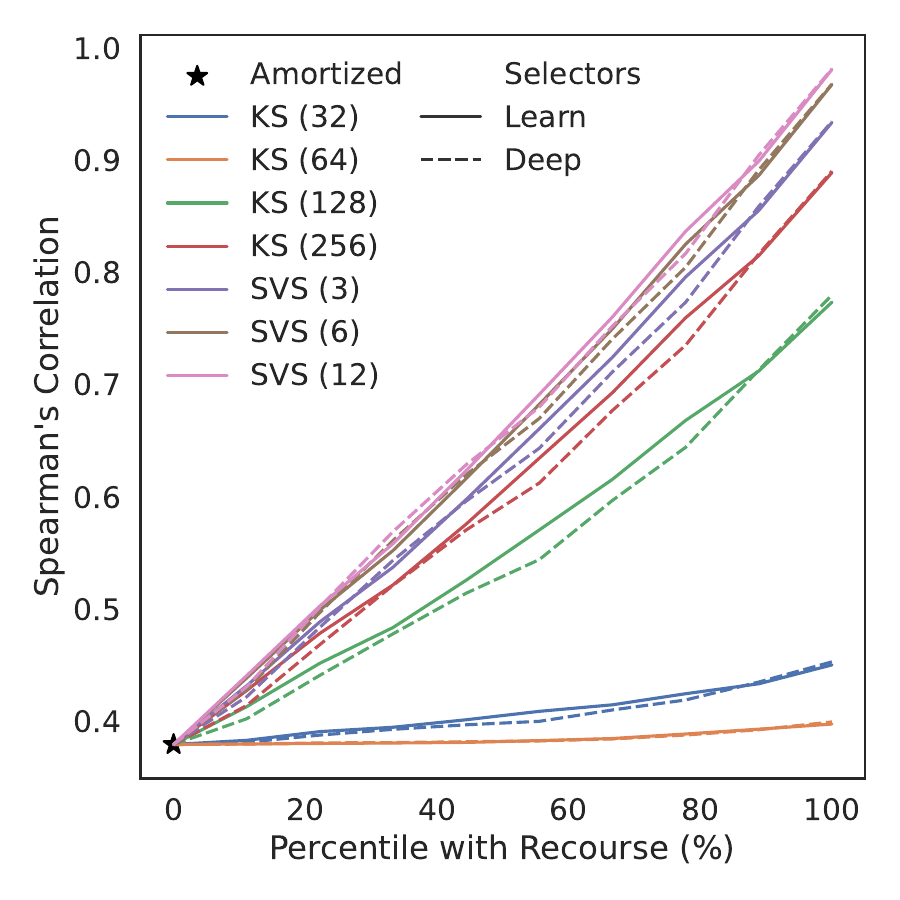}
    \caption{UCI-News}
  \end{subfigure}
  \begin{subfigure}[b]{0.24\linewidth}
    \includegraphics[width=\linewidth]{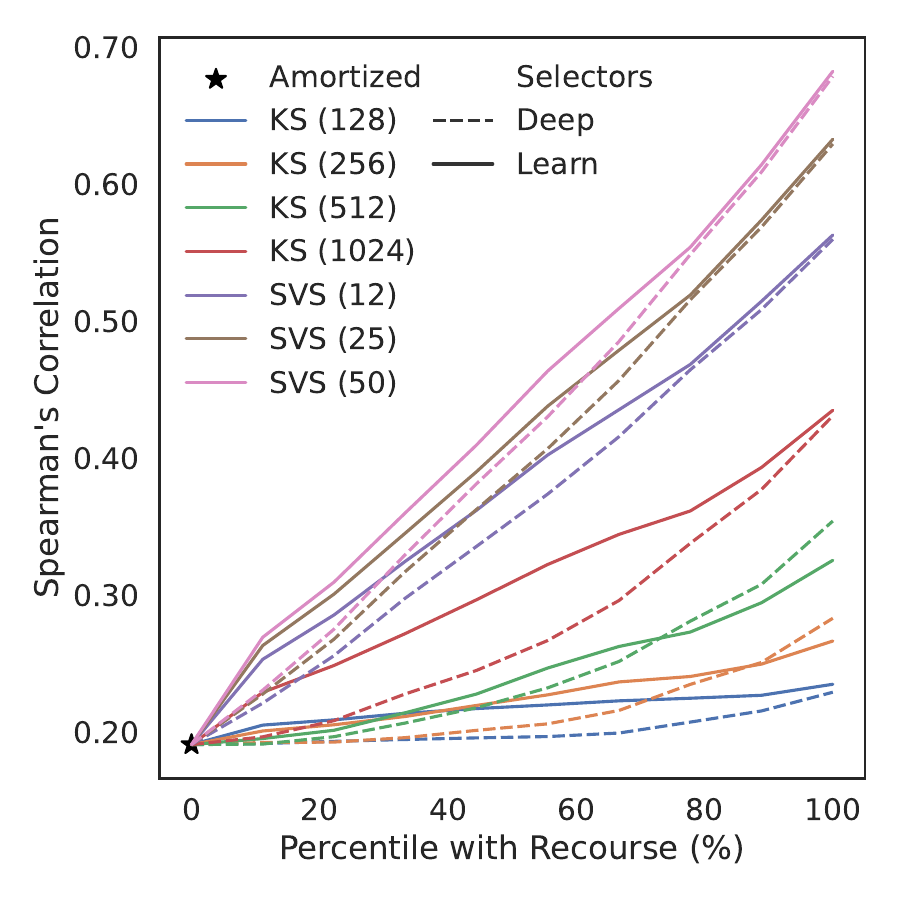}
    \caption{Yelp Review}
  \end{subfigure}
  \begin{subfigure}[b]{0.24\linewidth}
    \includegraphics[width=\linewidth]{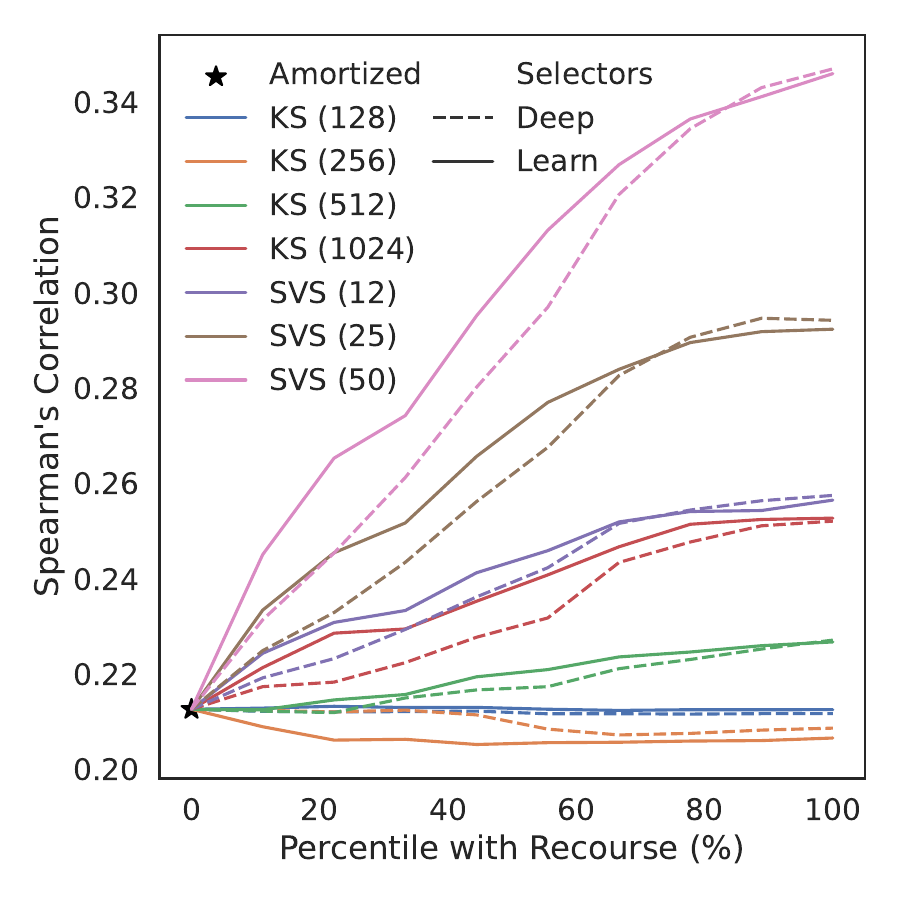}
    \caption{Toxigen}
  \end{subfigure}
  \caption{MSE (top) and Spearman's correlation (bottom) for selective explanations using different Monte Carlo  explainers.}
  \label{fig:various_methods}
\end{figure}

\subsection{Time Sharing Using Selective Explanations}
\label{apx:time_sharing}

\begin{figure}[H]
  \centering
  \begin{subfigure}[b]{0.4\linewidth}
    \includegraphics[width=\linewidth]{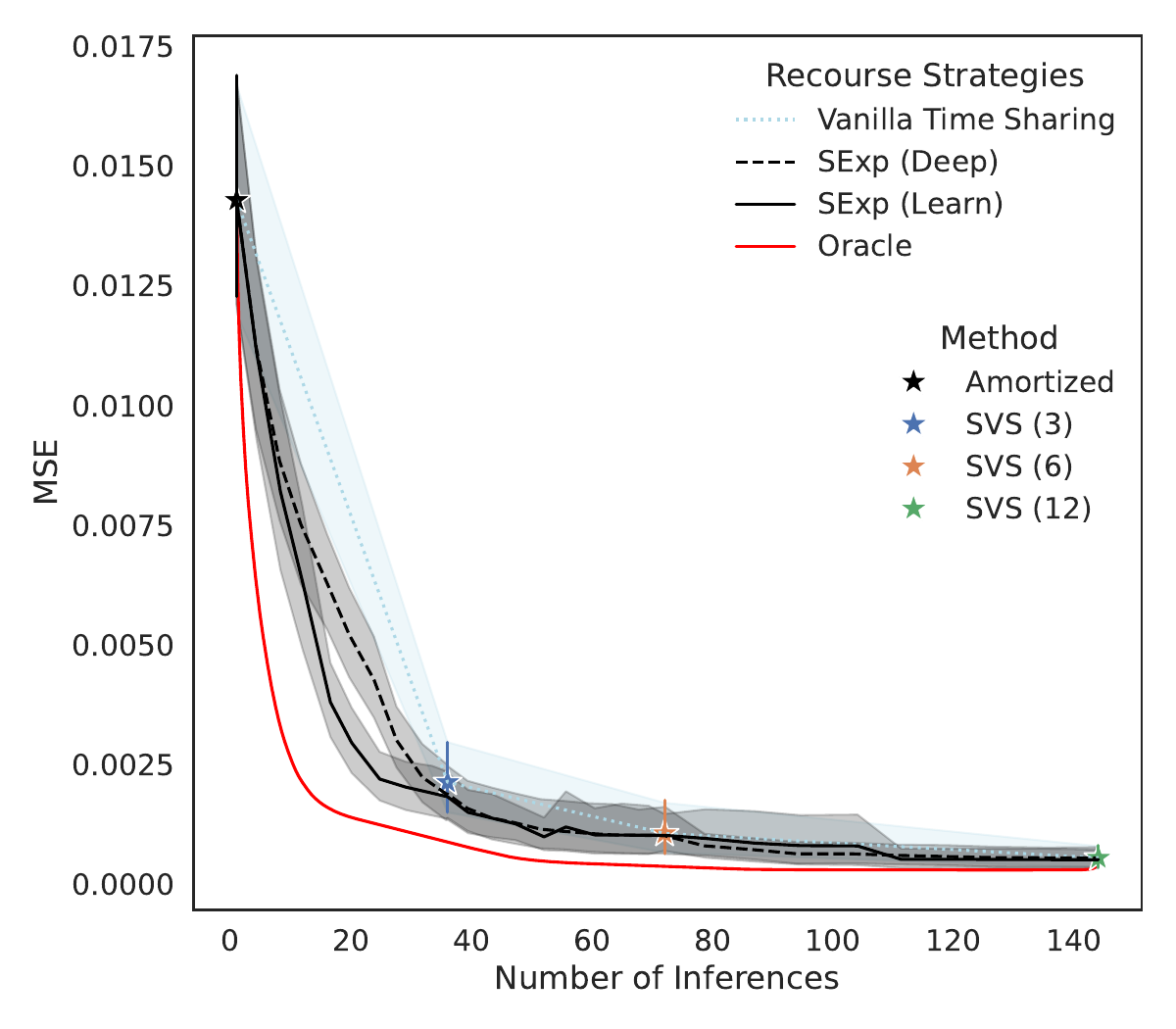}
     \caption{UCI-Adult}
  \end{subfigure}
  \begin{subfigure}[b]{0.4\linewidth}
    \includegraphics[width=\linewidth]{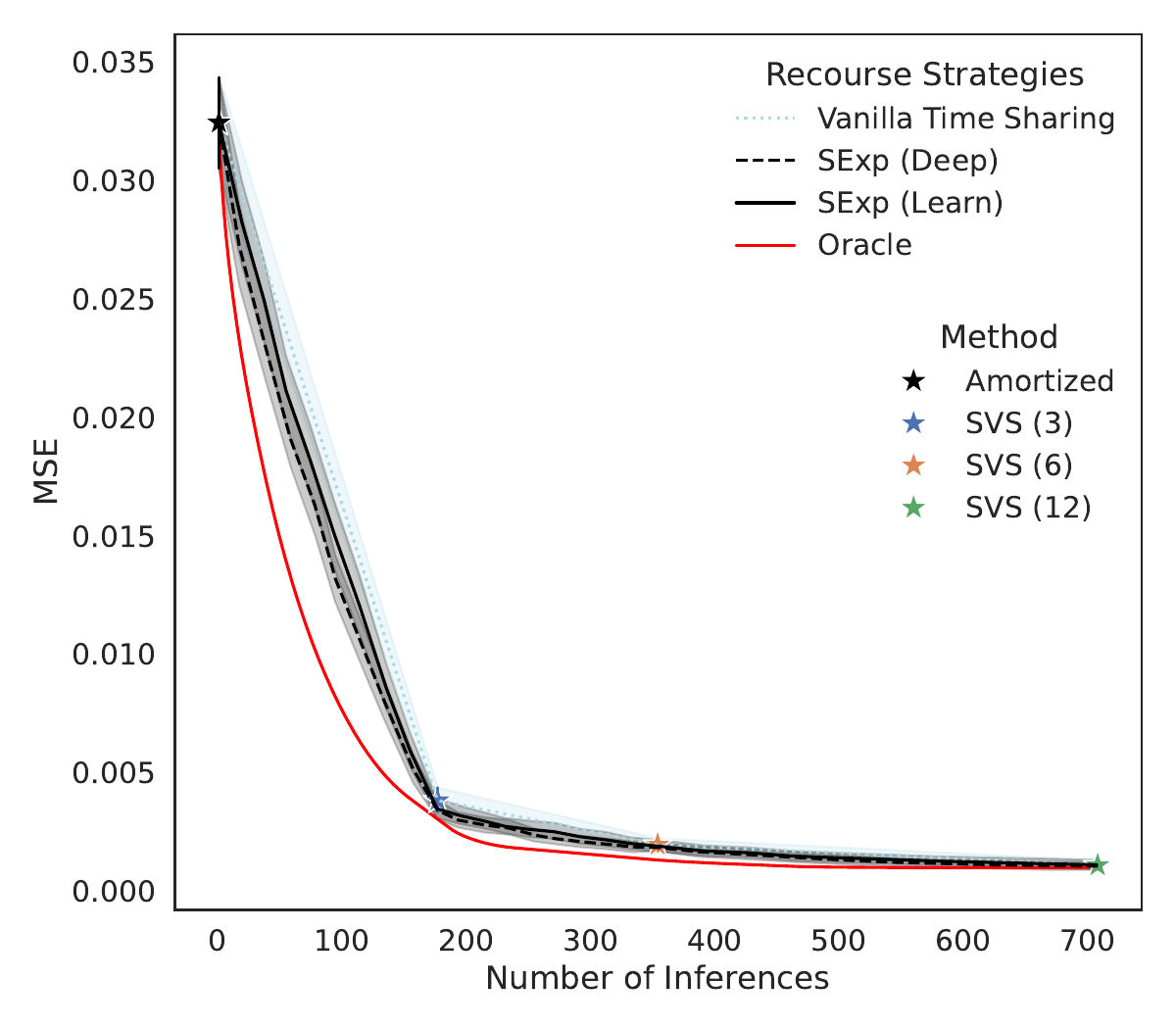}
     \caption{UCI-News}
  \end{subfigure}
  \begin{subfigure}[b]{0.4\linewidth}
    \includegraphics[width=\linewidth]{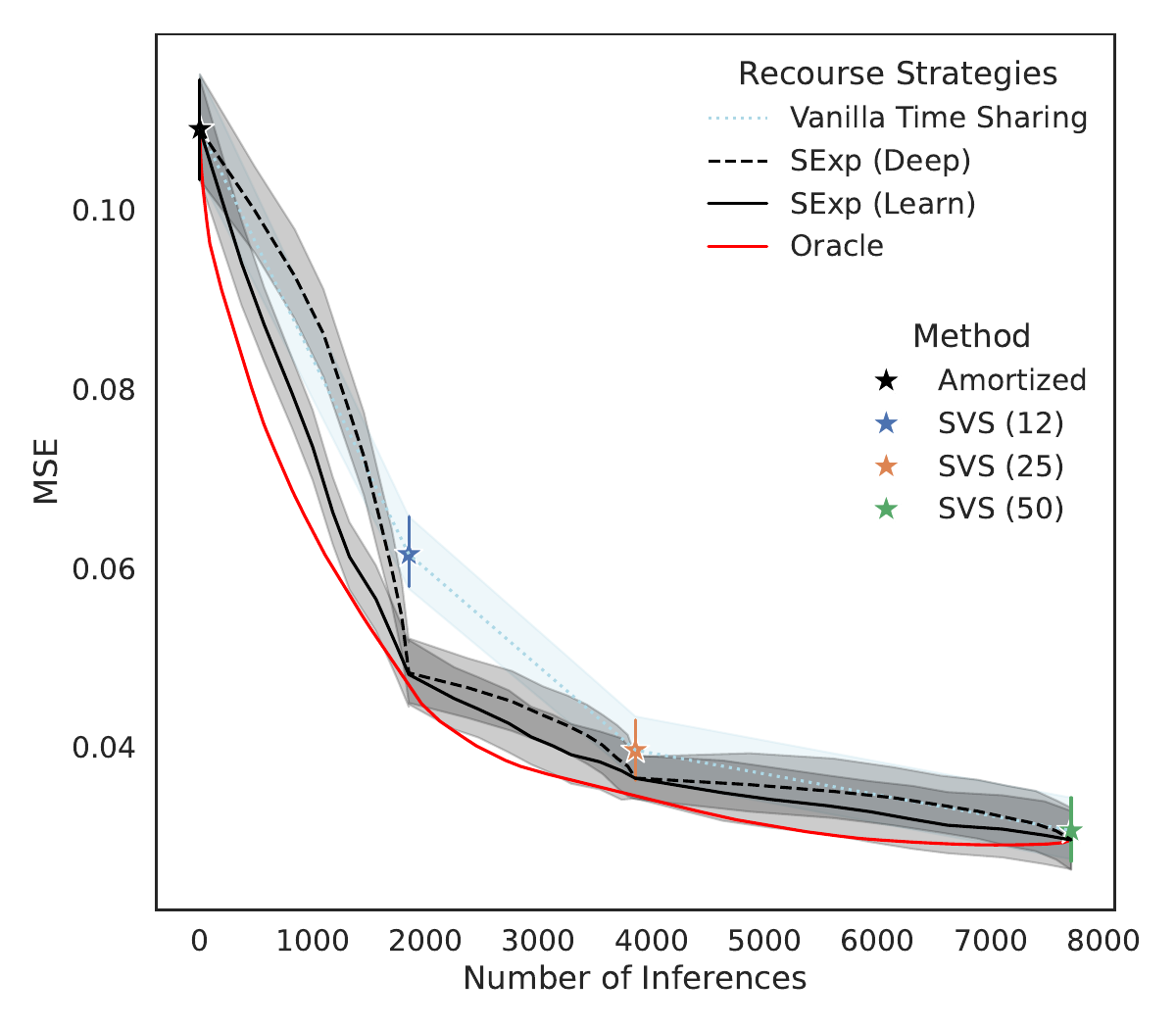}
     \caption{Yelp Review}
  \end{subfigure}
  \begin{subfigure}[b]{0.4\linewidth}
    \includegraphics[width=\linewidth]{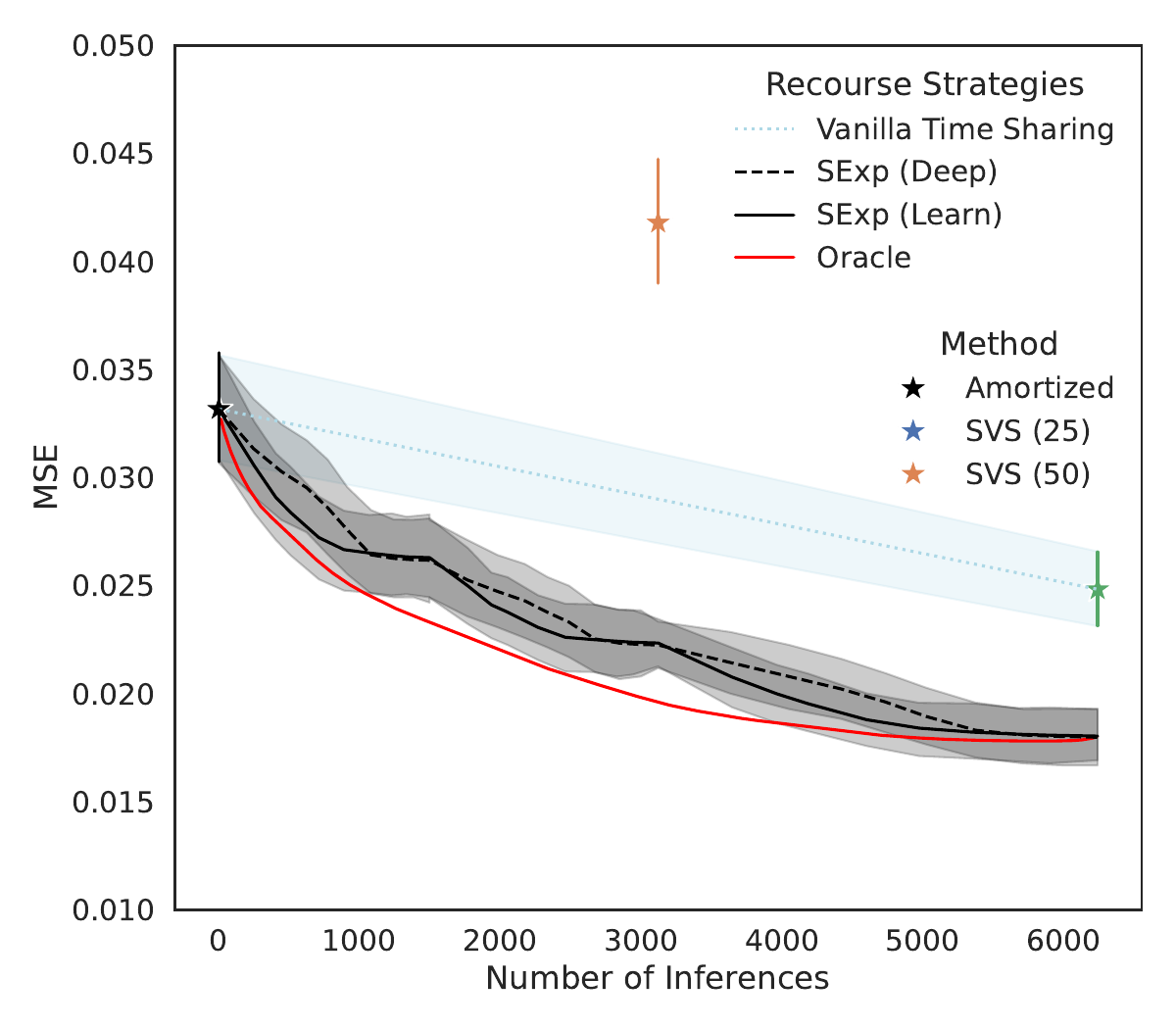 }
     \caption{Toxigen}
  \end{subfigure}
  \caption{Number of model inferences (x-axis) vs. MSE (y-axis) using (i) vanilla time sharing, (ii) time sharing using selective explanations compared to (iii) the oracle when the MSE of the provided explanation is known. }
\label{fig:time_sharing}
\end{figure}

We analyze how selective explanations can be used to improve the quality of Monte Carlo methods by time sharing between methods.
When computing explanations using Monte Carlo methods, we perform $n$ model inferences (x-axis in Figure \ref{fig:time_sharing}) until a desired MSE (y-axis in Figure \ref{fig:time_sharing}) is achieved.
This is done by gradually increasing the number of inferences per points we generate explanations -- this is displayed by the blue dotted curve in Figure \ref{fig:time_sharing} and we name it vanilla time sharing because the inferences (time) are shared gradually across points.
We also compare it with the Oracle given the red curve in Figure \ref{fig:time_sharing} where, for each point, we compute Monte Carlo explanations using SVS with parameter $12$, $25$, and $50$, compute their MSE to high-quality explanations and give the best explanation possible for a given number of inferences.
Oracle is the best that can be done in terms of MSE vs. Number of Inferences only using Monte Carlo explanations.
We compare both Orcle and vanilla time sharing with time sharing using selective explanations given by the black lines in  Figure \ref{fig:time_sharing}.
For the time sharing using selective explanations, we also gradually increase the number of inferences but use selective explanations instead of plain Monte Carlo explanations.

Figure \ref{fig:time_sharing} shows that selective explanations closely approximate the Oracle curve, indicating the selective explanations have close to optimal trade-off between the number of model inferences and MSE.
We highlight the performance of selective explanations in the Toxigen dataset. With only 1000 model inferences, we get better performance than using SVS-50 with about 6000 model inferences.
We also note that in both LLMs, using selective explanations closely approximates the oracle and provides a better explanation with the same number of inferences than just using SVS.

\section{Proofs of Theoretical Results}
\label{apx:proofs}

\begin{theorem}[Optimal $\lambda_h$]
Let $0=\alpha_1 < \alpha_2 < ... < \alpha_m = 1$ and define $Q_i$ as in \eqref{eq:quantile_definition}.
Then, $\lambda_i$ that solves the optimization problem in \eqref{eq:lambda_problem} is given by
\begin{equation}
    \lambda_i = \frac{\sum_{\substack{ (\bx, \by) \in \calD_{\texttt{val}} \\ \uncert(\bx) \in Q_i}} \langle \MC^{n}(\bx, \by) - \MC^{n'}(\bx, \by), \MC^{n}(\bx, \by) - \amortized(\bx, \by) \rangle}{\sum_{\substack{ (\bx, \by) \in \calD_{\texttt{val}} \\ \uncert(\bx) \in Q_i}} \left|\left| \amortized(\bx, \by) - \MC^n(\bx, \by)\right|\right|_2^2}.
\end{equation}
\end{theorem}

\begin{proof}
First, recall that 
\begin{align}
    \lambda_i &\triangleq \argmin_{\lambda \in \mathbb{R}} \sum_{\substack{ (\bx, \by) \in \calD_{\texttt{val}} \\ \uncert(\bx) \in Q_i}} \left|\left|  \selective(\bx, \by) - \MC^{n'}(\bx, \by) \right|\right|_2^2 \\
    &= \argmin_{\lambda \in \mathbb{R}} \sum_{\substack{ (\bx, \by) \in \calD_{\texttt{val}} \\ \uncert(\bx) \in Q_i}} \left|\left|  \lambda \amortized(\bx, \by) + (1 - \lambda)\MC^{n}(\bx, \by) - \MC^{n'}(\bx, \by) \right|\right|_2^2.
    \label{aeq:convexity}
\end{align}

Note that the function in \eqref{aeq:convexity} is convex in $\lambda$; therefore, if the derivative of it with respect to $\lambda$ is zero, then the lambda that achieves the zero gradient is the minima.
So, let's derivate \eqref{aeq:convexity} to find $\lambda_i$.

\begin{align}
    0 & = \frac{d}{d\lambda} \sum_{\substack{ (\bx, \by) \in \calD_{\texttt{val}} \\ \uncert(\bx) \in Q_i}} \left|\left|  \lambda \amortized(\bx, \by) + (1 - \lambda)\MC^{n}(\bx, \by) - \MC^{n'}(\bx, \by) \right|\right|_2^2 \\
    & = 2\sum_{\substack{ (\bx, \by) \in \calD_{\texttt{val}} \\ \uncert(\bx) \in Q_i}} \lambda || \MC^{n}(\bx, \by) - \amortized(\bx, \by) ||^2 \\ 
    & - 2\sum_{\substack{ (\bx, \by) \in \calD_{\texttt{val}} \\ \uncert(\bx) \in Q_i}} \langle \MC^{n}(\bx, \by) - \MC^{n'}(\bx, \by),  \MC^{n}(\bx, \by) - \amortized(\bx, \by)\rangle
\label{aeq:derivative}
\end{align}
From \eqref{aeq:derivative} we conclude the proof by showing that 
\begin{align}
     \lambda_i = \lambda & = \frac{\sum_{\substack{ (\bx, \by) \in \calD_{\texttt{val}} \\ \uncert(\bx) \in Q_i}} \langle \MC^{n}(\bx, \by) - \MC^{n'}(\bx, \by),  \MC^{n}(\bx, \by) - \amortized(\bx, \by)\rangle}{\sum_{\substack{ (\bx, \by) \in \calD_{\texttt{val}} \\ \uncert(\bx) \in Q_i}} || \MC^{n}(\bx, \by) - \amortized(\bx, \by) ||^2 }.
\end{align}
\end{proof}

\begin{theorem}[$\lambda_i \approx \lambda^{\texttt{opt}}_i$] 
Let the Monte Carlo explanation used to provide recourse $\MC^{n}$ to be different enough from the amortized explainer, i.e., $\EE{|| \MC^{n}(X, Y) - \amortized(X, Y) ||^2} = \mu > 0$.  Also, assume that $\MC^{n'}$ is a good Monte Carlo approximation for the high-quality explainer $\hq$, i.e.,  $\EE{|| \MC^{n'}(X, Y) - \hq(X, Y) ||^2} = \mu^*$ for $\epsilon > \frac{\sqrt{5\mu^*}}{\mu}$.
Recall that $\bx \in \mathbb{R}^d$.
If the explanations are bounded, i.e., $||\MC^{n}(\bx, \by)||, ||\amortized(\bx, \by)||, ||\hq(\bx, \by)| < Cd$ for some $C > 0$ then
\begin{equation}
    \Pr[|\lambda_i -\lambda^{\texttt{opt}}_i| > \epsilon] \leq e^{\frac{-\mu^2|Q_i|}{4Cd}} + e^{\frac{-\mu^4\epsilon^4|Q_i|}{400Cd}},
\end{equation}
where $|Q_i|$ is the number of points $\bx$ in the validation dataset $\calD_{\texttt{val}}$ that are in the bin $Q_i$.
\end{theorem}

\begin{proof}
Denote $  |Q_i| = | \{(\bx, \by) \in \calD_{\texttt{val}}, \text{ s.t. } \uncert(\bx) \in Q_i\}|$.

We start by showing that if $\EE{|| \MC^{n}(X, Y) - \amortized(X, Y) ||^2} = \mu$ then 
\begin{align}
    \ \ & \Pr\left[\frac{1}{|Q_i|} \sum_{\substack{ (\bx, \by) \in \calD_{\texttt{val}} \\ \uncert(\bx) \in Q_i}} || \MC^{n}(\bx, \by) - \amortized(\bx, \by) ||^2 \leq \frac{\mu}{2}\right] \\
    = & \Pr\left[\mu - \frac{1}{|Q_i|} \sum_{\substack{ (\bx, \by) \in \calD_{\texttt{val}} \\ \uncert(\bx) \in Q_i}} || \MC^{n}(\bx, \by) - \amortized(\bx, \by) ||^2  \geq \frac{\mu}{2}\right] \\
    \leq & e^{\frac{-\mu^2|Q_i|}{4Cd}}.
    \label{aeq:concentration_small_1}
\end{align}
Where the inequality in \eqref{aeq:concentration_small_1} follows from Hoeffding's inequality and the fact that:
\begin{equation}
    || \MC^{n}(\bx, \by) - \amortized(\bx, \by) ||^2 \leq ||\MC^{n}(\bx, \by)  || + ||  \amortized(\bx, \by)|| \leq 2Cd.
\end{equation}

Second, we recall that $\EE{|| \MC^{n'}(X, Y) - \hq(X, Y) ||^2} = \mu^* \leq \frac{\mu^2 \epsilon^2}{5}$. Then, we have that
\begin{align}
    \ \ & \Pr\left[\frac{1}{|Q_i|} \sum_{\substack{ (\bx, \by) \in \calD_{\texttt{val}} \\ \uncert(\bx) \in Q_i}} || \hq(\bx, \by) - \amortized(\bx, \by) ||^2 \geq \epsilon^2 \frac{\mu^2}{4} \right] \\
    = & \Pr\left[ \frac{1}{|Q_i|} \sum_{\substack{ (\bx, \by) \in \calD_{\texttt{val}} \\ \uncert(\bx) \in Q_i}} || \hq(\bx, \by) - \amortized(\bx, \by) ||^2 -  \mu^* \geq \epsilon^2 \frac{\mu^2}{4} -  \mu^*\right]\\
    \leq & \Pr\left[ \frac{1}{|Q_i|} \sum_{\substack{ (\bx, \by) \in \calD_{\texttt{val}} \\ \uncert(\bx) \in Q_i}} || \hq(\bx, \by) - \amortized(\bx, \by) ||^2 -  \mu^* \geq \epsilon^2 \frac{\mu^2}{20}\right] \\
    \leq & e^{\frac{-\mu^4\epsilon^4|Q_i|}{400Cd}}.
    \label{aeq:concentration_small_2}
\end{align}
Where the inequality in \eqref{aeq:concentration_small_2} follows from Hoeffding's inequality and the fact that:
\begin{equation}
    || \hq(\bx, \by) - \amortized(\bx, \by) ||^2 \leq ||\hq(\bx, \by)  || + ||  \amortized(\bx, \by)|| \leq 2Cd.
\end{equation}

Third, notice by directly applying Theorem \ref{thm:optimal_lambda} and replacing the Monte Carlo explanation by the high-quality explanation, we have that
\begin{equation}
    \lambda^{\texttt{opt}}_i = \frac{\sum_{\substack{ (\bx, \by) \in \calD_{\texttt{val}} \\ \uncert(\bx) \in Q_i}} \langle \MC^{n}(\bx, \by) - \hq(\bx, \by),  \MC^{n}(\bx, \by) - \amortized(\bx, \by)\rangle}{\sum_{\substack{ (\bx, \by) \in \calD_{\texttt{val}} \\ \uncert(\bx) \in Q_i}} || \MC^{n}(\bx, \by) - \amortized(\bx, \by) ||^2 }.
\end{equation}

Hence, we can write $\lambda^{\texttt{opt}}_i - \lambda_i$ as
\begin{align}
    \ \ &| \lambda^{\texttt{opt}}_i - \lambda_i |\\
    = & \left| \frac{\sum_{\substack{ (\bx, \by) \in \calD_{\texttt{val}} \\ \uncert(\bx) \in Q_i}} \langle \MC^{n'}(\bx, \by) - \hq(\bx, \by),  \MC^{n}(\bx, \by) - \amortized(\bx, \by)\rangle}{\sum_{\substack{ (\bx, \by) \in \calD_{\texttt{val}} \\ \uncert(\bx) \in Q_i}} || \MC^{n}(\bx, \by) - \amortized(\bx, \by) ||^2 } \right| \\
    \leq &  \frac{\left(\sum_{\substack{ (\bx, \by) \in \calD_{\texttt{val}} \\ \uncert(\bx) \in Q_i}} || \MC^{n'}(\bx, \by) - \hq(\bx, \by)||_2^2 || \MC^{n}(\bx, \by) - \amortized(\bx, \by)||_2^2 \right)^{1/2}}{\sum_{\substack{ (\bx, \by) \in \calD_{\texttt{val}} \\ \uncert(\bx) \in Q_i}} || \MC^{n}(\bx, \by) - \amortized(\bx, \by) ||^2 } , 
    \label{aeq:cauchy}
\end{align}
where the last inequality \eqref{aeq:cauchy} comes from the Cauchy–Schwarz inequality. Denote the denominator in \eqref{aeq:cauchy} by $\Delta$, i.e., 
$$\sum_{\substack{ (\bx, \by) \in \calD_{\texttt{val}} \\ \uncert(\bx) \in Q_i}} || \MC^{n}(\bx, \by) - \amortized(\bx, \by) ||^2 = \Delta.$$

Lastly, notice that $\MC^{n'}(\bx, \by)$ is sampled independently of $\MC^{n}(\bx, \by)$ and that $\hq(\bx, \by)$ is deterministic.
Therefore:
\begin{align}
    \ \ & \Pr[| \lambda^{\texttt{opt}}_i - \lambda_i | \geq \epsilon] \\
    \leq & \Pr\left[\frac{\left(\sum_{\substack{ (\bx, \by) \in \calD_{\texttt{val}} \\ \uncert(\bx) \in Q_i}} || \MC^{n'}(\bx, \by) - \hq(\bx, \by)||_2^2 || \MC^{n}(\bx, \by) - \amortized(\bx, \by)||_2^2 \right)^{1/2}}{\sum_{\substack{ (\bx, \by) \in \calD_{\texttt{val}} \\ \uncert(\bx) \in Q_i}} || \MC^{n}(\bx, \by) - \amortized(\bx, \by) ||^2 } 
  \geq \epsilon   \right] \label{aeq:applying_cauchy_to_prob} \\
 \leq & \Pr\left[\frac{\sum_{\substack{ (\bx, \by) \in \calD_{\texttt{val}} \\ \uncert(\bx) \in Q_i}} || \MC^{n'}(\bx, \by) - \hq(\bx, \by)||_2^2 || \MC^{n}(\bx, \by) - \amortized(\bx, \by)||_2^2 }{\Delta^2} 
 \geq \epsilon^2   \right] \\
  \leq & \Pr\left[ \left. \frac{\sum_{\substack{ (\bx, \by) \in \calD_{\texttt{val}} \\ \uncert(\bx) \in Q_i}} || \MC^{n'}(\bx, \by) - \hq(\bx, \by)||_2^2 || \MC^{n}(\bx, \by) - \amortized(\bx, \by)||_2^2 }{\Delta^2} \geq \epsilon^2  \right|  \Delta \leq \frac{\mu}{2}\right] \nonumber \\
  & \times \Pr\left[ \Delta \leq \frac{\mu}{2}\right] \nonumber\\
  + & \Pr\left[ \left. \frac{\sum_{\substack{ (\bx, \by) \in \calD_{\texttt{val}} \\ \uncert(\bx) \in Q_i}} || \MC^{n'}(\bx, \by) - \hq(\bx, \by)||_2^2 || \MC^{n}(\bx, \by) - \amortized(\bx, \by)||_2^2 }{\Delta^2} \geq \epsilon^2  \right|  \Delta > \frac{\mu}{2}\right] \nonumber  \\
  & \times \Pr\left[ \Delta > \frac{\mu}{2}\right] \label{aeq:conditioning}\\
  \leq & \Pr\left[  {\sum_{\substack{ (\bx, \by) \in \calD_{\texttt{val}} \\ \uncert(\bx) \in Q_i}} || \MC^{n'}(\bx, \by) - \hq(\bx, \by)||_2^2 || \MC^{n}(\bx, \by) - \amortized(\bx, \by)||_2^2 } \geq \epsilon^2 \frac{\mu^2}{4}\right] \nonumber \\
  & + \Pr\left[ \Delta \leq \frac{\mu}{2}\right] \label{aeq:ind_and_bounded}\\
  \leq & e^{\frac{-\mu^2|Q_i|}{4Cd}} + e^{\frac{-\mu^4\epsilon^4|Q_i|}{400Cd}}. \label{aeq:final_equation}
\end{align}
Where the inequality in \eqref{aeq:applying_cauchy_to_prob} is a direct application of \ref{aeq:cauchy}, the inequality in \eqref{aeq:conditioning} comes from simply conditioning, the inequality in \eqref{aeq:ind_and_bounded} comes from the fact that probabilities are bounded by one getting rid of the first term in \eqref{aeq:ind_and_bounded} (first out of lines) and the fourth term in \eqref{aeq:ind_and_bounded} (forth out of lines) and the fact that $\MC^{n'}(\bx, \by)$ is sampled independently of $\MC^{n}(\bx, \by)$ and that $\hq(\bx, \by)$ is deterministic. Finally, the last inequality in \eqref{aeq:final_equation} comes from applying \eqref{aeq:concentration_small_1} and \eqref{aeq:concentration_small_2}.

Hence, from \eqref{aeq:final_equation}, we conclude that
\begin{equation}
    \Pr[| \lambda^{\texttt{opt}}_i - \lambda_i | \geq \epsilon] \leq e^{\frac{-\mu^2|Q_i|}{4Cd}} + e^{\frac{-\mu^4\epsilon^4|Q_i|}{400Cd}}.
\end{equation}

\end{proof}

\begin{proposition}[Coverage for Inference Budget]
Let $\text{N}_{\texttt{budget}} \geq 1$ be the set inference budget, and assume that the Monte Carlo method $\MC^n(\bx, \by)$ uses $n$ model inferences.
Then, the coverage level $\alpha$ should be chosen such that
\begin{equation}
    \argmin_{\alpha \in [0, 1]} \left\{ \EE{\text{N}(\selective(\bx, \by))} \leq \text{N}_{\texttt{budget}} \right\} = \frac{n + 1 - \text{N}_{\texttt{budget}}}{n}.
\end{equation}
Recall that Shapley Value Sampling with parameter $m$ performs $1 + dm$ inferences ($\bx \in \mathbb{R}^d$), and Kernel Shap with parameter $m$ performs $m$ inferences.
\end{proposition}

\begin{proof}
    Let $\alpha \in [0, 1]$, then an $\alpha$ portion of examples receive explanations from the amortized explainer, i.e., they receive one inference, and $1 - \alpha$ portion of examples receive explanations with initial guess, i.e., $n$ model inferences. Therefore, the expected number of model inferences per instance is given by \eqref{aeq:average_inferences}.
    \begin{equation}
        \EE{\text{N}(\selective(\bx, \by))} = \alpha + (1 - \alpha)(n+1)
        \label{aeq:average_inferences}
    \end{equation}
    In order for the inference budget to be followed, it is necessary that
    \begin{equation}
        \EE{\text{N}(\selective(\bx, \by))} = \alpha + (1 - \alpha)(n+1) \leq \text{N}_{\texttt{budget}}.
        \label{aeq:average_inferences_bound}
    \end{equation}
    From $\eqref{aeq:average_inferences_bound}$, we conclude that:
    \begin{equation}
        \alpha \geq \frac{n + 1 - \text{N}_{\texttt{budget}}}{n},
        \label{aeq:average_inferences_algebra}
    \end{equation}
    Hence, 
    \begin{equation}
        \argmin{\alpha \in [0, 1]} \left\{ \EE{\text{N}(\selective(\bx, \by))} \leq \text{N}_{\texttt{budget}} \right\} = \frac{n + 1 - \text{N}_{\texttt{budget}}}{n}.
    \end{equation}
\end{proof}

\end{document}